\documentclass[12pt]{article}

\usepackage{subcaption}

\newif\ifcomments

\usepackage[letterpaper, margin=1in]{geometry}

\newcommand{\paragraphNoSkip}[1]{\paragraph{#1.}}
\newcommand{\captionDescription}[1]{\caption{#1}}

\usepackage[backend=biber, style=numeric, sorting=nyt, maxbibnames=10, maxalphanames=10, minalphanames=10]{biblatex} 

\let\localcite\cite
\renewcommand{\cite}[1]{\localcite{#1}}
\let\localciteauthor\citeauthor
\newcommand{\citet}[1]{\localciteauthor{#1}~\localcite{#1}}

\usepackage[T1]{fontenc}
\usepackage[english]{babel}
\usepackage{csquotes}
\usepackage{tabularx}
\usepackage{array}

\newcolumntype{P}{>{\raggedright\arraybackslash}X}

\usepackage{mathtools}
\usepackage{amsmath}
\usepackage{amssymb}
\newcommand{\cmark}{\checkmark}
\usepackage{utfsym}

\newcommand{\cxmark}{\cmark\kern-1.1ex\raisebox{.7ex}{\rotatebox[origin=c]{125}{--}}}

\usepackage{multirow}
\usepackage{amsthm}
\usepackage{thmtools,thm-restate}
\usepackage[ruled,linesnumbered]{algorithm2e}
\usepackage{enumitem}
\usepackage{graphicx}
\graphicspath{{./figures/}}

\usepackage{tikz}
\usepackage{pgfplots}
\usepackage{pgf-pie}
\pgfplotsset{compat=1.18}
\usetikzlibrary{dateplot}
\newcommand*\shortyear[1]{\expandafter\@gobbletwo\number\numexpr#1\relax} 
\usetikzlibrary{patterns}

\definecolor{c0}{RGB}{165 0 38}
\definecolor{c1}{RGB}{215 48 39}
\definecolor{c2}{RGB}{244 109 67}
\definecolor{c3}{RGB}{253 174 97}
\definecolor{c4}{RGB}{254 224 144}
\definecolor{c5}{RGB}{255 255 191}
\definecolor{c6}{RGB}{224 243 248}
\definecolor{c7}{RGB}{171 217 233}
\definecolor{c8}{RGB}{116 173 209}
\definecolor{c9}{RGB}{69 117 180}
\definecolor{c10}{RGB}{49 54 149}

\usepackage{booktabs}
\usepackage{csvsimple}
\usepackage[most]{tcolorbox}

\usepackage{url}
\PassOptionsToPackage{hyphens}{url}

\usepackage{hyperref}
\PassOptionsToPackage{breaklinks}{hyperref}

\usepackage{xcolor}

\usepackage[capitalise]{cleveref}

\newtheorem*{example*}{Example}

\crefname{definition}{Definition}{Definitions}
\crefname{theorem}{Theorem}{Theorems}
\crefname{corollary}{Corollary}{Corollaries}
\crefname{example}{Example}{Examples}
\crefname{remark}{Remark}{Remarks}

\usepackage[symbols,acronym,nonumberlist,nogroupskip,stylemods={mcols,longbooktabs}]{glossaries-extra}
\makenoidxglossaries
\glssetcategoryattribute{acronym}{nohyper}{true}
\setabbreviationstyle[acronym]{long-short}

\newacronym[longplural={Markov processes}]{MP}{MP}{Markov process}
\newacronym[longplural={Markov decision processes}]{MDP}{MDP}{Markov decision process}
\newacronym{AI}{AI}{artificial intelligence}
\newacronym{AMM}{AMM}{automated market maker}
\newacronym{APY}{APY}{annual percentage yield}
\newacronym{HAC}{HAC}{Heteroskedasticity and Autocorrelation Consistent}
\newacronym{APR}{APR}{annual percentage rate}
\newacronym{ASIC}{ASIC}{application specific integrated circuit}
\newacronym{CDF}{CDF}{cumulative density function}
\newacronym{CPU}{CPU}{central processing unit}
\newacronym{DAA}{DAA}{difficulty-adjustment algorithm}
\newacronym{DQL}{DQL}{deep-Q-learning}
\newacronym{DeFi}{DeFi}{decentralized finance}
\newacronym{EIP}{EIP}{Ethereum improvement proposal}
\newacronym{ERC}{ERC}{Ethereum request for comments}
\newacronym{EVM}{EVM}{Ethereum virtual machine}
\newacronym{LP}{LP}{liquidity provider}
\newacronym{LT}{LT}{liquidity taker}
\newacronym{LTV}{LTV}{loan-to-value}
\newacronym{MEV}{MEV}{Maximal Extractable Value}
\newacronym{ML}{ML}{machine learning}
\newacronym{OO}{OO}{order optimization}
\newacronym{PDF}{PDF}{probability density function}
\newacronym{PID}{PID}{proportional integral derivative}
\newacronym{PoS}{PoS}{proof-of-stake}
\newacronym{PoW}{PoW}{proof-of-work}
\newacronym{RAM}{RAM}{random-access memory}
\newacronym{RL}{RL}{reinforcement learning}
\newacronym{RPC}{RPC}{remote procedure call}
\newacronym{SSD}{SSD}{solid state drive}
\newacronym{URL}{URL}{uniform resource locator}
\newacronym{USD}{USD}{United States Dollar}
\newacronym{WETH}{WETH}{Wrapped Ethereum}
\newacronym{WBTC}{WBTC}{Wrapped Bitcoin}
\newacronym{block-DAG}{block-DAG}{block directed-acyclic-graph}
\newacronym{geth}{geth}{Go Ethereum}
\newacronym{p2p}{p2p}{peer to peer}
\newacronym{FaaS}{FaaS}{front-running-as-a-service}
\newacronym{CEX}{CEX}{centralized exchange}
\newacronym{DEX}{DEX}{decentralized exchange}
\newacronym{TVL}{TVL}{total value locked}
\newacronym{CFMM}{CFMM}{constant function market maker}
\newacronym{TFM}{TFM}{transaction fee mechanism}
\newacronym{IC}{IC}{incentive compatible}
\newacronym{FN}{FN}{false name}
\newacronym{MIC}{MIC}{miner incentive compatible}
\newacronym{wrt}{w.r.t.}{with regards to}
\newacronym{QoS}{QoS}{Quality-of-Service}
\newacronym{DoS}{DoS}{Denial-of-Service}
\newacronym{NFT}{NFT}{non fungible token}
\newacronym{KYC}{KYC}{know your customer}
\newacronym{UX}{UX}{user experience}
\newacronym{L1}{L1}{layer 1 blockchain}
\newacronym{L2}{L2}{layer 2 blockchain}
\newacronym{iid}{i.i.d.}{independent and identically distributed}
\newacronym{wlog}{w.l.o.g.}{without loss of generality}
\newacronym{ENS}{ENS}{Ethereum Name Service}
\newacronym{LST}{LST}{liquid staking token}
\newacronym{SoK}{SoK}{systemization of knowledge}
\newacronym{DRL}{DRL}{deep reinforcement learning}
\newacronym{ECDF}{ECDF}{empirical cumulative distribution function}
\newacronym{PPO}{PPO}{Proximal Policy Optimization}
\newacronym{VDF}{VDF}{verifiable delay function}
\newacronym{DRB}{DRB}{distributed randomness beacon}

\glsxtrnewsymbol[description={The underlying token of the blockchain.}]{token}{\ensuremath{\leavevmode\clap{\hskip7pt--}\mathcal{T}}}
\newcommand{\token}{{\gls[hyper=false]{token}}}

\glsxtrnewsymbol[description={The \gls{LST}.}]{lst}{\ensuremath{\leavevmode\clap{\hskip5pt--}\mathcal{L}}}
\newcommand{\lst}{{\gls[hyper=false]{lst}}}

\glsxtrnewsymbol[description={A block.}]{block}{\ensuremath{b}}

\glsxtrnewsymbol[description={The cost of $1\lst$, denominated in $\token$.}]{cost}{\ensuremath{c}}
\newcommand{\cost}{{\gls[hyper=false]{cost}}}

\begin{document}
\title{Your Loss is My Gain:\\Low Stake Attacks on Liquid Staking Pools}

\author{
    Sen Yang$^{*,1}$
    \and
    Aviv Yaish$^{*,1,2}$
    \and
    Arthur Gervais$^{3}$
    \and
    Fan Zhang$^{1}$
}
\date{
    \small
    $^*$Equal contribution\\
    $^1$Yale University \& IC3\\
    $^2$Complexity Science Hub, Vienna\\
    $^3$University College London \& UC Berkeley RDI
}

\maketitle
\begin{abstract}
Permissionless Proof-of-Stake (PoS) economic security is predicated on the high cost of violating consensus safety or liveness.
We show that liquid staking introduces additional risks that are not captured by standard PoS economic security arguments.
Through an empirical study of Ethereum data, we find that the operational performance of liquid staking pools is positively associated with subsequent normalized liquid staking token (LST) returns.
Motivated by this, we present a cross-layer attack: a low-stake adversary can manipulate the consensus protocol to degrade a target pool's performance and take application-layer positions that profit if the market reprices the corresponding \gls{LST} in-line with the historically observed association.

To make the consensus layer manipulation concrete, we develop a deep reinforcement learning (DRL) framework to automatically discover attack strategies.
Our evaluation shows that the learned strategies can recover near-optimal theoretical attacks and uncover new manipulation behaviors that significantly degrade target pool performance.
We further characterize feasible application-layer monetization channels and analyze leveraged shorting in detail using Monte Carlo simulations, showing that such attacks can be profitable with over one-half probability for LSTs of major staking pools.
Our findings reveal a previously overlooked attack surface in PoS systems with liquid staking and expose a gap between consensus and economic security.

\end{abstract}

\section{Introduction}
Permissionless \Gls{PoS} blockchains allow users to become validators, i.e., take part in processing transactions, by locking tokens as illiquid stake.
To prevent an overproliferation of validators, blockchains can set a minimum stake amount, thus also lowering communication overhead.
However, some worry that this entry barrier mays harm decentralization~\cite{buterin2024possible, buterin2024some};
consider that Ethereum, the most popular \gls{PoS} blockchain, mandates a minimal stake of $32$ ETH, currently worth over $100{,}000$ USD.
Market forces addressed this concern by (1) forming staking \emph{pools} where users (short of the minimal amount or not) collectively stake their funds, and (2) issuing tradeable \emph{\glspl{LST}} to pool members which can be redeemed for their share of the stake.
Pooled and liquid staking have become Ethereum's main staking methods, consistently comprising about $50\%$ of stake since Ethereum's transition to \gls{PoS} (see \cref{figure:StakeDistributionCategory,figure:StakeOverTime}).
Given this prominence, it is natural to ask:
\begin{quote}
\emph{Do staking pools and liquid staking contribute to the security of the system?}
\end{quote}

\begin{figure}
    \centering
    \begin{tikzpicture}[scale=1,align=center]
    \pie[
        text=pin,
        radius=1.8,
        explode=0.1,
        color={c0!40, c4!40, c10!40}
    ]{
        49.0/Liquid staking\\pools,
        27.9/Others,
        23.1/CEXs
    }
    \end{tikzpicture}
    \captionDescription{The distribution of stake in Ethereum by staker category, based on data from~\cite{dune2026eth2staking} as of Feb 2026, is shown. As is evident, liquid staking is the leading staking category.}
    \label{figure:StakeDistributionCategory}
\end{figure}

\subsection{This work}
To answer this question, we present attacks that target liquid staking pools and allow an adversary with a small fraction of the stake to increase its profit.
Our attacks work \emph{even} when the underlying consensus mechanism is provably economically secure.
Such mechanisms typically ensure that the cost of violating safety (honest nodes agree on the identity and order of processed transactions) and liveness (new transactions are eventually processed) is higher than the possible profits for attackers without a large fraction of stake (e.g., over $33\%$).
Our attacks remain viable even in the face of such strong guarantees due to several key ideas:
(1) Instead of violating safety or liveness to generate a profit, attacks can influence staking market dynamics to the benefit of the adversary;
(2) Market dynamics are affected by pool performance: profit-maximizing users go with the best-performing pools, and the value of a pool's \glspl{LST} is positively associated with its performance, using the pool's total \gls{APR}, i.e., the realized annualized return delivered to stakers;
(3) Pool performance can be manipulated by adversaries inside and outside the pool.

\paragraphNoSkip{Understanding market dynamics}
To motivate our attacks, we perform the first analysis of how staking pool performance affects \gls{LST} prices.
Using each pool's recent staking \gls{APR}, we find that better performance is consistently associated with subsequent movements in normalized \gls{LST} prices across six liquid staking pools, with {\em moderate correlations} for most pools and {\em strong correlations} for Coinbase and Binance.
For example, Lido shows a moderate Pearson correlation of 0.391, while Coinbase shows a strong Pearson correlation of 0.645.

\paragraphNoSkip{Cross-layer attacks}
Motivated by this observation, we analyze a class of cross-layer attacks against liquid staking pools which separate the \emph{attack surface} from the \emph{profit channel}: the adversary manipulates consensus-layer performance, but extracts profit at the application layer.
As illustrated in \Cref{fig:attack-overview}, under observed market dynamics, the adversary can monetize the attack through application-layer actions (e.g., taking a short position) that make a profit when \gls{LST} prices decrease.

\begin{figure}[th]
    \centering
    \includegraphics[width=0.65\linewidth]{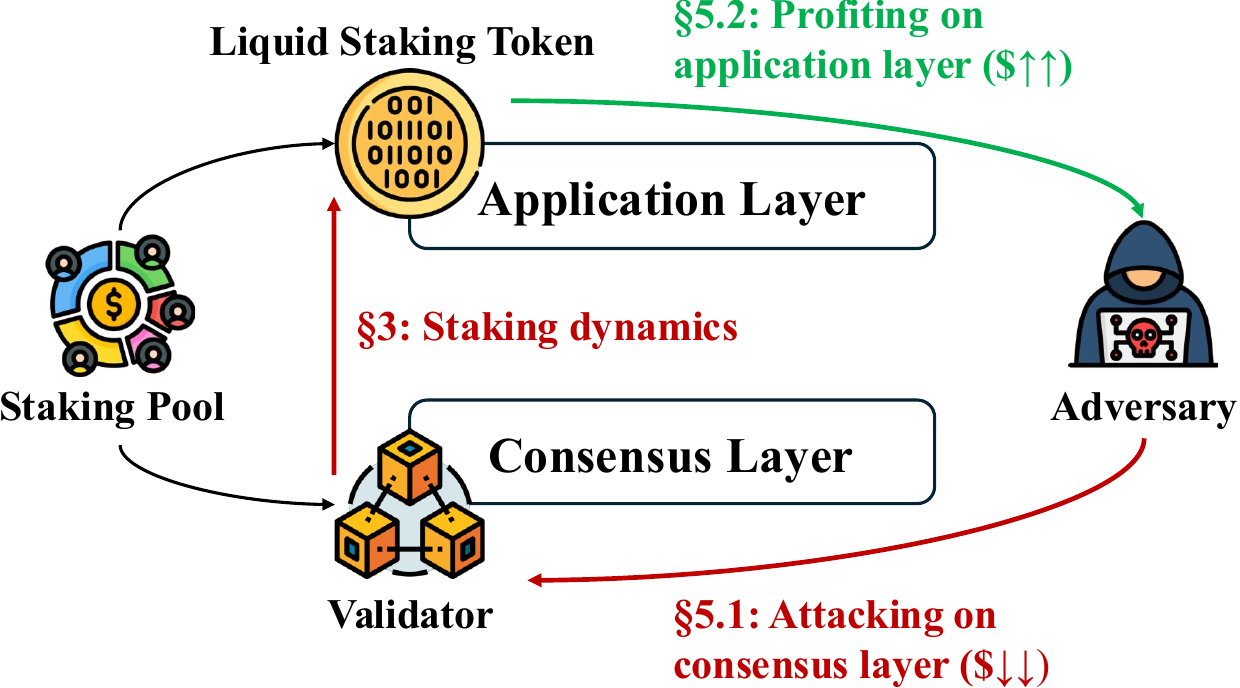}
    \captionDescription{An overview of our attacks. The adversary first targets the consensus layer by attacking validators controlled by a liquid staking pool, which degrades the pool's performance.
    If the corresponding \gls{LST} is repriced according to the target pool's performance (as analyzed in \Cref{section:EmpiricalData}), the price decrease can be monetized via financial positions.}
    \label{fig:attack-overview}
\end{figure}

To make these attacks concrete, we first focus on a simple and realistic class of consensus-layer attacks available to a low-stakes adversary: manipulation of the \gls{PoS} leader-election process~\cite{ferreira2022optimal,alpturer2024optimal}. We develop a \gls{DRL} framework that automatically discovers such strategies, allowing the adversary to bias the mechanism, reduce a victim pool's realized block allocation, and thereby degrade its performance.
Through carefully designed reward functions, our framework can steer the adversary to either maximize its own block allocation or, when desired, trade off its own allocation to minimize the victim's realized block allocation.
We further extend the environment to model heterogeneous block values induced by \gls{MEV}, allowing us to study how \gls{MEV} opportunity costs affect learned manipulation strategies.
We use self-allocation maximization mainly as a validation baseline and focus on victim-targeted degradation.

We then turn to the application layer.
Given the observed \gls{LST} price relationship with pool performance, we provide a taxonomy of monetization strategies. As a representative case, we analyze leveraged shorting, which is simple to execute and supported by deep liquidity, to estimate the adversary's potential profits and resulting attack incentives.

\paragraphNoSkip{Evaluation} 
We evaluate the framework across three goals.

First, as a baseline validation, we configure the reward to maximize the adversary's block allocation and show that learned policies recover known self-optimization strategies~\cite{alpturer2024optimal,nagy2025forking}.

Second, we change the objective to victim-targeted degradation, which differs from prior work that focused on increasing the adversary's allocation.
Under this objective, the learned policies trade-off limited adversarial loss for disproportionate victim loss.
For example, an adversary with 20\% stake can reduce the realized block allocation of another 20\% stake pool by about 8.3\%.
These degradation-oriented policies also impose a more significant impact on chain quality: adversary more often sacrifice or displace blocks to harm the victim pool.

Third, we evaluate an MEV-aware extension with heterogeneous block values.
We find that MEV changes the learned behavior: self-optimization policies become more conservative about sacrificing the adversary's high-value slots, while self-optimization and victim-targeted policies become more willing to fork-displace victim blocks with valuable MEV opportunities.

Finally, using Monte Carlo simulations calibrated to real-world data and parameters, we find that leveraged shorting can yield a likelihood of profit above one-half for three representative pools: Coinbase, Ether.fi, and Rocket Pool. Compared with the roughly $3\%$ APR from honest staking, the expected profit from successful executions could be $11\times$ higher, providing a strong incentive even when attacks are executed infrequently.

\paragraphNoSkip{Contributions}
We summarize our contributions below.
\begin{itemize}[leftmargin=*]
    \item We are the first to empirically study the relation between liquid staking pool performance and \gls{LST} prices; our analysis finds consistent positive associations, with moderate correlations for several pools and a strong correlation for Coinbase and Binance.
    \item Building on the empirical link, we identify a class of cross-layer attacks that separate the attack surface from the profit channel: the adversary degrades a target pool's consensus-layer performance and monetizes \gls{LST} price movements through application-layer strategies.
    \item We implement a \gls{DRL} framework for discovering \gls{PoS} leader-election manipulation strategies under different objectives. It recovers known self-optimization strategies as a validation baseline and identifies degradation-oriented strategies that impose disproportionate loss on a specific victim pool.
    \item We propose a comprehensive taxonomy of application-layer monetization channels and analyze shorting as a representative case; Monte Carlo simulations with real-world protocol parameters show that such attacks can be economically attractive.
\end{itemize}

\section{Background}
\label{sec:background}

\subsection{Terminology}
\paragraphNoSkip{Proof-of-Stake}
Proof-of-Stake (PoS) blockchains rely on stake-weighted participation to maintain consensus safety and liveness~\cite{kiayias2017ouroboros,daian2019snow,buterin2020combining,neu2021ebb,zhang2025available,yang2025geographical}.
To participate in consensus, an entity must lock a protocol-defined amount of capital as {\em stake} to become a {\em validator}, which gives it both the right and obligation to perform consensus duties such as block proposal and attestation; in return, validators earn protocol-defined rewards for doing so correctly.

\paragraphNoSkip{Liquid Staking}
Liquid staking allows users to participate in \gls{PoS} systems without directly operating validators.
Instead, users deposit their stake into a liquid staking pool, which delegates the stake to a set of operators that run validators on behalf of the pool. In return, the pool issues a liquid staking token (LST) that represents a claim on the underlying staked assets and future consensus rewards.
Liquid staking tokens typically follow two accounting designs: some staking pools distribute rewards via \textit{rebasing} token balances, aiming to maintain an approximate 1:1 peg to the underlying asset (e.g., ETH), whereas others keep token balances fixed and reflect rewards through token price appreciation.
For simplicity, unless stated otherwise, we use the term \gls{LST} to refer to {\em non-rebasing} \gls{LST}.

\subsection{Related work}
\label{section:RelatedWork}

\paragraphNoSkip{Staking pools and \glspl{LST}}
As far as the authors are aware, no prior empirical work systematically studies the correlation between staking pool performance and the \gls{LST} prices.
Meng~\cite{meng2024analysis} analyzes drivers of staking inflows and withdrawals in Ethereum.
Two empirical studies~\cite{xiong2024leverage, xiong2024exploring} analyze leveraged staking, and price discrepancies between primary and secondary \gls{LST} markets and related arbitrage frictions, respectively, while Scharnowski and Jahanshahloo~\cite{scharnowski2024economics} examine the liquid staking basis, its determinants, and its contribution to price discovery.
In parallel,~\cite{gersbach2022staking} proposes a game-theoretic model of staking pool formation that examines how reward-sharing rules shape incentives and security.

\paragraphNoSkip{Consensus-layer manipulation in \gls{PoS}}
Chen and Micali~\cite{chen2019algorand} present the Algorand \gls{PoS} protocol and bound the biasability of its leader selection mechanism. Ferreira et al.~\cite{ferreira2022optimal} construct explicit biasing attacks and quantify the adversary's advantage, with tighter bounds in~\cite{ferreira2024computing}. Cai et al.~\cite{cai2024profitable} further show that, under certain timing assumptions, profitable attacks on Algorand are always detectable.

A parallel line of work analyzes the vulnerability of Ethereum's leader election mechanism to similar biasing attacks.
Buterin~\cite{buterin2016validator,buterin2018randao} observes that RANDAO-based randomness beacons allow actors to manipulate outcomes by choosing whether or not to reveal their inputs.
Dworzanski~\cite{dworzanski2018note} formalizes this as a last-mover advantage, where only the final contributor can perfectly predict the effect of its action on the beacon.
Alpturer and Weinberg~\cite{alpturer2024optimal} derive optimal biasing strategies across stake levels, showing that a $20\%$ attacker can increase its block share to $20.68\%$.
Nagy et al.~\cite{nagy2025forking} further show that biasing strategies can be combined with forking, expanding the adversary's action space and increasing the achievable block share to $20.96\%$.
Both works formulate leader-election manipulation as an explicit \gls{MDP} and derive strategies for increasing the adversary's block share via dynamic programming.
In this work, we instead adopt a \gls{DRL} approach that supports multiple manipulation objectives, including maximizing the adversary's own block share and minimizing the victim pool's realized block share; our main focus is the latter.
A broad selection of randomness beacon protocols are surveyed and compared in \cite{choi2023sok}.

A complementary line of work studies strategic forking in Ethereum PoS, including ex-ante reorg attacks~\cite{neuder2021low,schwarz2022three,sarenche2025commitment}, balancing attacks~\cite{neu2021ebb,neu2022two,schwarz2022three}, bouncing attacks~\cite{pavloff2023ethereum}, justification-withholding attacks~\cite{potuz2024justification}, and staircase attacks~\cite{zhang2024max,li2025bunnyfinder}.
Most of these attacks have since been mitigated by protocol updates~\cite{zhang2025available}. In principle, they could also be adapted for consensus-layer manipulation, but typically require a much larger stake share.

\paragraphNoSkip{Pool-insider attacks}
This line of work studies attacks initiated by pool participants rather than external adversaries.
Rosenfeld~\cite{rosenfeld2011analysis} identifies a block-withholding attack against PoW mining pools, where an adversarial participant earns pool rewards by submitting ``locally valid'' blocks while withholding ``globally valid'' blocks.
For staking pools, Tzinas and Zindros~\cite{tzinas2024principal} show that a legitimate node operator can use standard participation privileges to degrade a pool's performance without changing consensus rules.
Similar ideas can also support our cross-layer attacks, as discussed in \Cref{sec:other-consensus-layer-attacks}.

\paragraphNoSkip{Shorting-based attacks}
Prior work theoretically discussed profiting from price drops~\cite{ford2019rationality,tzinas2024principal}.
We go further in several ways.
First, we empirically analyze the association between pool performance and token value.
Second, we provide a manipulation vector for a concrete blockchain protocol, which moreover does not require the adversary to operate the victim's node.
We follow with a taxonomy of available monetization channels, and then estimate an adversary's potential profit via Monte Carlo simulations calibrated using real-world data.

\section{Empirical Pool Performance-LST Price Relationship}
\label{section:EmpiricalData}

We now empirically study the correlation between liquid staking pool performance and associated \gls{LST} prices using real-world data. We focus on Ethereum, the largest \gls{PoS} system with the most active liquid staking ecosystem.

\subsection{Data collection}
To measure staking dynamics, we collect four types of data: (1) \gls{LST} token prices; (2) liquid staking pool tags; (3) the execution-layer and consensus-layer rewards of liquid staking pools; and (4) the effective balances of liquid staking pools. We collect these data from block $15{,}537{,}393$ to block $22{,}677{,}606$, covering the period from the Merge (September 15, 2022)~\cite{buterin2024possible} to June 10, 2025, which corresponds to $1{,}000$ days under Ethereum \gls{PoS}.

\paragraphNoSkip{ETH and \gls{LST} prices}
We collect historical USD-denominated prices of ETH and \gls{LST} tokens from CoinGecko~\cite{coingecko2025api}.
The dataset covers daily prices of six \glspl{LST}, corresponding to six liquid staking pools: Lido, Binance, Coinbase, Ether.fi, Rocket Pool, and Mantle, which collectively include both large and smaller liquid staking pools (see \Cref{sec:staking-pool-distribution} for details).

\paragraphNoSkip{Liquid staking pool tags}
We collect validator public keys and pool tags to map validators to liquid staking pools.
We begin with a dataset that maps Ethereum validators to specific entities~\cite{dune2026eth2staking}, covering $1{,}969{,}253$ validators, of which $1{,}887{,}716$ have pool tags.
We then supplement this dataset with an address-to-pool-tag mapping based on prior work~\cite{nagy2025forking}, which includes $65$ validator addresses associated with $43$ liquid staking pools.

\paragraphNoSkip{Liquid staking pool performance}
For each pool, we measure performance using its {\em \gls{APR}}, a standard annualized measure of staking return.
Intuitively, \gls{APR} summarizes the realized return generated by the pool's validators over a given period, scaled by stake and annualized, and thus captures the end-to-end outcome of the pool's staking operations.

To compute \gls{APR} of each liquid staking pool, we combine rewards from both the consensus layer and the execution layer.
For the consensus layer, we obtain each validator's realized consensus rewards by querying a local Ethereum node via its consensus-layer API~\cite{ethereum2026beaconapi}. 
Specifically, we collect per-epoch attestation rewards, proposer rewards, and sync committee rewards for each validator and aggregate them based on liquid staking pool tags.

For the execution layer, we distinguish between blocks built through MEV-Boost~\cite{flashbots2026mevboost} and blocks built locally.
If a block is built through  MEV-Boost, we obtain the validator's reward from the corresponding MEV-Boost relay API~\cite{flashbots2026mevboostrelay}.
Otherwise, we follow prior work~\cite{yang2025decentralization,oz2024who} and compute the block value directly from on-chain block.
We aggregate these execution-layer rewards using the validator addresses associated with each liquid staking pool.

Finally, we collect daily effective balances of validators associated with each pool from the consensus-layer API and aggregate it to obtain pools' daily effective balances.
We compute a pool's daily staking return as the sum of its consensus-layer and execution-layer rewards divided by its daily effective balance, and annualize this quantity to obtain the pool's \gls{APR}.

\subsection{Empirical analysis}
\paragraphNoSkip{Price normalization}
Since \gls{LST} prices can move with market dynamics, comparing raw USD-denominated prices can mix in a common market factor.
To better study the correlation between pool performance and the corresponding \gls{LST} price ($P_{\mathrm{LST}}$), we normalize each \gls{LST} price by the ETH price ($P_{\mathrm{ETH}}$).
Specifically, we use $P_{\mathrm{LST}}/P_{\mathrm{ETH}}$ as our price measure.

\paragraphNoSkip{Relation between performance and future returns}
To assess whether staking performance is incorporated into \gls{LST} valuation, we use a unified specification that relates contemporaneous performance to subsequent token returns.
We focus on 180-day forward log returns rather than same-day prices, as same-day co-movement does not distinguish whether price changes follow performance or merely occur concurrently.
A forward-return specification better captures whether performance is gradually reflected in market valuation.

To avoid selective sample choices, we use each pool's full available sample and apply the same specification across pools.
Specifically, we report Pearson and Spearman correlations, and regressions of 180-day forward log returns on contemporaneous pool APR.
Standard errors are computed using \gls{HAC}~\cite{andrews1991heteroskedasticity} adjustment to account for serial correlation and heteroskedasticity in daily price series.

\begin{table}[t]
\centering
\captionDescription{Analysis results for the relationship between pool performance and \gls{LST} prices.
Pearson captures linear association, and Spearman captures rank-based monotone association.
Correlations of 0.3-0.5 are moderate, and those above 0.5 are strong~\cite{cohen2013statistical}.
Larger $\beta$ values indicate stronger return sensitivity to APR.
Smaller \gls{HAC} Standard Errors (HACSE) relative to $\beta$ indicate more precise estimates.
Smaller $p$-values indicate stronger evidence.
}
\label{tab:lst-price-performance}
\scalebox{0.95}{
\begin{tabular}{lccccc}
\toprule
Pool & Pearson & Spearman & $\beta$ & HAC SE & $p$-value \\
\midrule
Lido           & 0.391 & 0.387 & 0.1049 & 0.0204 & <0.0001  \\
Binance        & 0.434 & 0.516 & 0.1738 & 0.0270 & <0.0001  \\
Coinbase       & 0.645 & 0.708 & 0.8719 & 0.1359 &  <0.0001 \\
Ether.fi       & 0.407 & 0.302 & 1.7247 & 0.7692 & 0.0250 \\
Rocket Pool     & 0.176 & 0.260 & 0.1142 & 0.0407 &  0.0050  \\
Mantle & 0.112 & 0.315 & 0.1151 & 0.0550 & 0.0365 \\
\bottomrule
\end{tabular}
}
\end{table}

As \Cref{tab:lst-price-performance} shows, both Pearson and Spearman correlations are positive across all six pools, suggesting that better pool performance is generally associated with higher subsequent \gls{LST} returns.
The association is strongest for Coinbase, with Pearson and Spearman correlations of 0.645 and 0.708, respectively, indicating a strong positive relation under common effect-size guidelines~\cite{cohen2013statistical}.
Binance shows a relatively strong association, with correlations of 0.434 and 0.516, while the remaining pools show positive associations of varying magnitudes.

The regression coefficient $\beta$ measures how strongly contemporaneous pool APR associates with future 180-day forward log returns.
All regression coefficients are positive, with the largest estimates observed for Ether.fi and Coinbase.
Together with the \gls{HAC} standard errors and corresponding $p$-values, the results are suggestive of a consistent positive association between higher pool performance and higher future returns.

Taken together, these results show that \gls{LST} prices and pool performance are meaningfully linked, providing an empirical foundation for our subsequent analysis of cross-layer attacks.
Robustness checks with 90-day and 60-day horizons yield qualitatively similar positive relations across pools, as reported in \Cref{sec:robustness-of-price-performance-analysis}.

\section{Model}
\label{section:Model}
We now formally define our model.
Notations are summarized in \cref{section:Glossary}.

We consider a \gls{PoS} consensus mechanism $\mathcal{M}$ with a validator set $\mathcal{V}$, which is partitioned into two disjoint subsets: a strategic adversary $\mathcal{A}$ and a set of honest validators $\mathcal{H}$, such that $\mathcal{V} = \mathcal{A} \cup \mathcal{H}$.
The adversary $\mathcal{A}$ controls cumulative relative stake $\alpha_{\mathcal{A}} \in (0, 1)$, while honest validators $\mathcal{H}$ possess the remaining stake $1 - \alpha_{\mathcal{A}}$.
As with previous work, stake distribution is constant within the scope of an attack.
Finally, we consider adversaries with sufficient capital to execute application-layer financial strategies, e.g., take short positions.

The mechanism $\mathcal{M}$ is defined by two components that govern the evolution of the chain:
\begin{itemize}[leftmargin=*]
\item \textbf{Leader Election ($\mathcal{M}_L$):} This component defines the proposer selection process. Ideally, a validator $v \in \mathcal{V}$ is selected with probability proportional to its stake; intuitively, it determines \emph{who is allowed to propose a block}.
\item \textbf{Fork Choice Rule ($\mathcal{M}_F$):} This determines the blockchain's canonical head by resolving potential conflicts between multiple branches according to a rule that aggregates validators' staking-weighted signals; intuitively, it determines \emph{which proposed block becomes canonical}.
\end{itemize}

\paragraphNoSkip{Example: Ethereum} 
Take Ethereum as an example. The leader election component $\mathcal{M}_L$ in Ethereum is driven by a randomness beacon based on RANDAO~\cite{edgington2023upgrading}.
Time in Ethereum is divided into epochs, each consisting of $32$ slots. In each slot, a single validator is selected as the block proposer, and each valid proposed block contains a 96-byte value, called \texttt{randao\_reveal}. This value is contributed by the proposer and incorporated into a running randomness mix. The mix is updated incrementally across slots and acts as a source of randomness for the protocol. The RANDAO mix at the end of an epoch $e$ is then used as an input for proposer selection and committee assignment in subsequent epochs, together with validators' effective stake balances.
Strictly speaking, this affects epoch $e+2$ in Ethereum, but for simplicity we refer to it as affecting epoch $e+1$ in the remainder of the paper; this convention does not alter any of our results.

The fork choice rule $\mathcal{M}_F$ in Ethereum is implemented by LMD-GHOST~\cite{buterin2020combining}. Validators issue attestations that vote for blocks, and the canonical head of the chain is determined as the block (together with its ancestor chain) that accumulates the largest staking-weighted attestations, considering only the latest message from each validator.
To prevent ex-ante reorg attacks~\cite{schwarz2022three}, proposer boost $p_{\text{boost}}$ is introduced to temporarily increase the effective weight of the proposer's block during fork choice, ensuring that a newly proposed block is favored over competing blocks that have accumulated adversarial attestations.
In the current specification, $p_{\text{boost}}$ is set to 40\% of the total effective stake balance.

\section{Our Attacks}
\label{section:Attacks}

Liquid staking pools span both the consensus layer, where their validators earn rewards through block proposals, and the application layer, where their \glspl{LST} are traded and used as financial assets.
This structure allows an adversary to attack one layer while monetizing the effect in another, making cross-layer attacks a general class of vulnerabilities in liquid staking pools.
In this section, we show how such attacks can be constructed by combining consensus-layer manipulation policies with application-layer financial positions.

\subsection{Consensus-layer manipulation}

\subsubsection{Manipulation landscape}
\label{sec:other-possible-attacks}

As discussed in \Cref{section:RelatedWork}, an adversary can leverage various approaches to degrade the target pool's performance. We summarize these approaches and their corresponding adversary stake requirements in \Cref{tab:attack-taxonomy}.

\begin{table}[th]
\centering
\captionDescription{Stake requirements for different strategies.}
\label{tab:attack-taxonomy}
\small
\begin{tabular}{l c}
\toprule
\textbf{Attack strategy} & \textbf{Adversary stakes} \\
\midrule
Leader election manipulation~\cite{alpturer2024optimal, nagy2025forking} & Low-High \\
Reorganization attacks~\cite{neu2022two,zhang2024max} & Medium-High  \\
Insider attacks~\cite{rosenfeld2011analysis,tzinas2024principal} & Low \\
Network attacks~\cite{heilman2015eclipse} & Zero \\
\bottomrule
\end{tabular}
\end{table}

While all these strategies can degrade the victim pool's performance, for example, reorganization attacks can invalidate proposed blocks or attestations via strategic forking, they typically have different requirements on the adversary's stake share and guarantees.
Reorganization attacks can be effective, but they often require a relatively large stake share to succeed. 
By contrast, insider and network-layer attacks are harder to quantify and lack clear performance guarantees. The former involves an adversary sabotaging the pool from within as a participant, whereas the latter relies on disrupting connectivity to hinder operational correctness.
We provide more details on these attack vectors in \Cref{sec:other-consensus-layer-attacks}.

In this section, we focus on consensus-layer manipulation via leader election, because prior work shows that even low-stakes adversaries can bias leader election~\cite{alpturer2024optimal,nagy2025forking}, and such manipulation has more predictable implications.

\subsubsection{Straw-man solution and limitation}
Manipulating the leader election mechanism allows the adversary to degrade a target liquid staking pool's performance by reducing its realized share of block proposals relative to its ideal, stake-proportional allocation.

A natural approach to modeling such manipulations is by casting the leader election and fork choice process as an \gls{MDP}.
Formally, an \gls{MDP} is a tuple $(S, A, P, R)$, where $S$ and $A$ denote the states and actions available to the agent, while $P$ and $R$ define the transition model and rewards.
Specifically, the probability of transitioning from state $s \in S$ to state $s' \in S$ when taking action $a \in A$ is given by $P_a(s, s') = \Pr(s_{t+1} = s' \mid s_t = s, a_t = a)$, and the agent's reward for the transition is given by $R_a(s, s')$.
A policy $\pi: S \to A$ is a function that specifies the action taken by the agent in each state.

Previous research~\cite{alpturer2024optimal,nagy2025forking} has modeled manipulation strategies targeting Ethereum's leader election and fork choice rules as an \gls{MDP}, with this approach being standard for consensus attacks \cite{eyal2014majority,gervais2016security,yaish2022blockchain,yaish2023uncle,bar2025mad}.
In this context, the adversary seeks to derive a policy that maximizes its long-term expected utility, specifically by augmenting its relative share of block proposals at the expense of honest participants.
This can be performed through the application of classical techniques, such as policy or value iteration~\cite{puterman2014markov}.

However, this approach also has a well-known limitation: the explosion of state space~\cite{hou2021squirrl}.
To model the attack accurately, a state $s \in \mathcal{S}$ must at least capture the 256-bit RANDAO mix, the proposer allocation for the current and look-ahead epochs, and the relative fork status (i.e., private vs. public chain heights).
In particular, \cite{nagy2025forking} shows that even in a single-adversary setting, the state space can grow rapidly, making an explicit MDP formulation infeasible.
This limitation motivates us to adopt a learning-based approach that does not require explicit enumeration of the state space. In particular, prior work has shown that deep reinforcement learning can be effective in studying selfish mining strategies~\cite{hou2021squirrl,barzur2023werlman,sarenche2025bitcoin}, suggesting its suitability for problems with a similar structure.

\subsubsection{Deep reinforcement learning}
In this paper, we adopt a \gls{DRL} approach to study consensus-layer manipulation against a target liquid staking pool. Our \gls{DRL} formulation differs from prior work in two ways. Methodologically, it avoids explicit state-space enumeration and instead learns policies directly from interaction with the simulated protocol environment. Substantively, it optimizes for targeted degradation of a victim pool's realized performance, rather than only for the adversary's own proposer gains.

\Cref{fig:rl-framework} presents the overview of our \gls{DRL} framework.
First, according to the consensus $\mathcal{M}$, we build an environment that simulates the execution of its leader election and fork choice components.
Then, the adversary is modeled as an RL agent that controls only its own validator actions, while honest validators follow the consensus mechanism $\mathcal{M}$ and act according to the protocol without learning or adaptation.
At each slot, the agent observes the system state, takes an action, and receives a reward that reflects the effect of its strategy on both adversarial gain and honest-validator loss under protocol-compliant execution.

\begin{figure}
    \centering
    \includegraphics[width=0.65\linewidth]{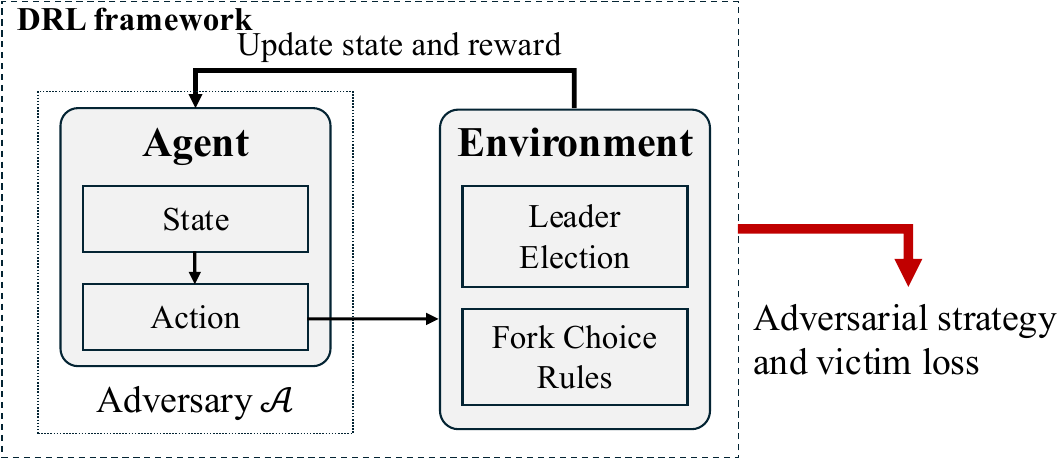}
    \captionDescription{\gls{DRL} framework for automatic discovery of consensus-layer adversarial strategies.
    }
    \label{fig:rl-framework}
\end{figure}

\paragraphNoSkip{Environment}
The environment serves as a formal abstraction of the consensus layer, specifically isolating the operational logic of the leader election mechanism and the fork choice rule (i.e., LMD-GHOST).
To maintain computational tractability, we assume a synchronous network with instantaneous broadcasting.
Within this framework, the environment dictates the system's evolution: it assigns block proposers for each slot and evaluates the canonical chain according to the prescribed fork choice rule.
The environment acts as the reactive entity that processes the adversary's actions to determine the next state and the corresponding reward $R$, reflecting the collective protocol response under these idealized conditions.

\paragraphNoSkip{Action space $A$}
For Ethereum, we follow prior work~\cite{nagy2025forking} to define a discrete action space $A$, which consists of the following actions in a given slot.
(1) \textbf{Propose ($a_{\text{prop}}$):} the agent honestly publishes a valid block extending the current public canonical head.
(2) \textbf{Miss ($a_{\text{miss}}$):} the agent intentionally skips its assigned proposal, thus not contributing entropy to the beacon.
(3) \textbf{Hide ($a_{\text{hide}}$):} the agent builds a block privately without broadcasting it to the network, withholding its entropy from public view.
(4) \textbf{Fork ($a_{\text{fork}}$):} the agent broadcasts a previously hidden private chain to trigger a reorganization (reorg), which has two effects: excluding honest blocks (particularly those by the target pool $\mathcal{T}$) from the canonical chain, and altering their entropy contributions to the beacon.
(5) \textbf{Regret ($a_{\text{regret}}$):} the agent abandons its private fork and resynchronizes with the public canonical chain.

\begin{figure}[htbp]
    \centering
    \includegraphics[width=0.65\linewidth]{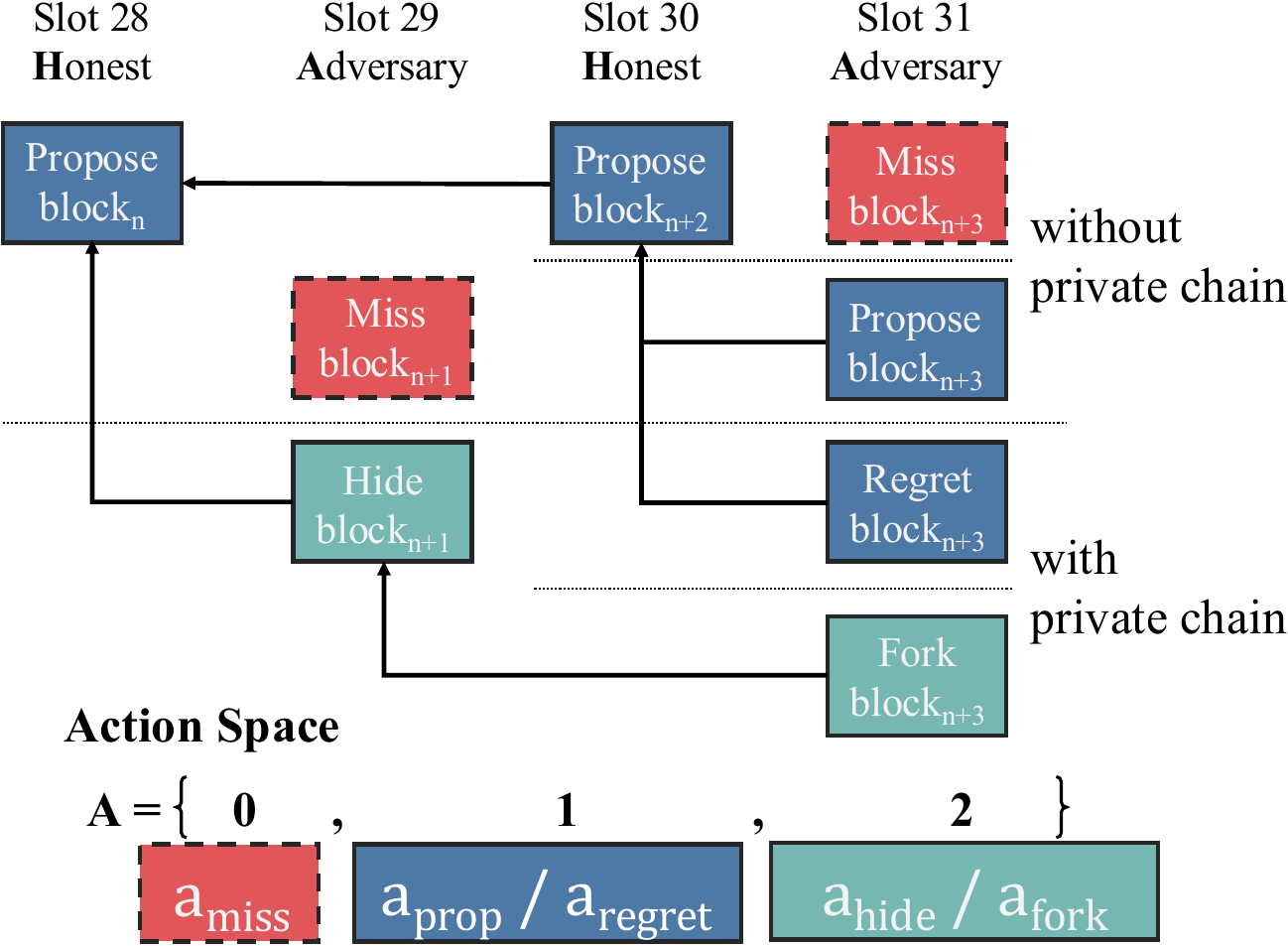}
    \captionDescription{The action space of the adversary.}
    \label{fig:action-space}
\end{figure}

While the protocol allows for these actions, we employ action consolidation to a compact space $A \in \{0, 1, 2\}$.
Specifically, $a = 0$ corresponds to $a_{\text{miss}}$, in which the agent intentionally skips its block proposal.
Action $a = 1$ represents honest behavior: when no private chain exists, it maps to $a_{\text{prop}}$; otherwise, it maps to $a_{\text{regret}}$, in which the adversary abandons the private chain and proposes honestly.
Action $a = 2$ corresponds to private chain manipulation: when no private chain exists, it maps to $a_{\text{hide}}$; when one exists and a competing honest block is observed, it maps to $a_{\text{fork}}$, initiating an ex-ante reorganization to displace the honest block.

This reduction follows from protocol feasibility constraints and the timing of meaningful adversarial decisions.
Once the adversary initiates a private chain, it must repeatedly execute the action  $a_{\text{hide}}$ to accumulate sufficient attestation weight; deviating earlier is infeasible if a reorganization is to succeed.
As a result, the adversary's only non-trivial decision point arises after a private chain exists and an honest validator publishes a competing block, at which moment the adversary chooses whether to adopt the action $a_{\text{fork}}$ or $a_{\text{regret}}$.

\paragraphNoSkip{State representation and observation design}
To make learning tractable, we transform the high-dimensional environment state $S$ into a compact, 7-dimensional observation vector $\phi(s) \in \mathcal{O}$, which can be categorized into three functional groups:
\begin{itemize}[leftmargin=*]
    \item \textit{Branch:} A binary indicator signaling the presence of a withheld private chain, which defines the agent's permitted actions by determining the feasibility of $a_{fork}$ or $a_{regret}$.
    \item \textit{Counterfactual strategic oracles:} A key difficulty in leader manipulation is that rewards are sparse and delayed: the effect of an action is only realized in the block allocation of the subsequent epoch. To provide the agent with denser, action-specific feedback, we augment the observation space with counterfactual oracles. An oracle helps the adversary estimate the expected block allocation in the next epoch under different hypothetical actions at the current slot. To make these estimates tractable, when evaluating a hypothetical action at a given slot, we assume that all subsequent adversarial slots in the remaining tail of the current epoch follow the default honest strategy (i.e., propose the block).

    Concretely, we define three oracles $\{\Omega_{prop}, \Omega_{skip}, \Omega_{fork}\}$. Each oracle returns two numbers: the expected block allocation for the adversary $\mathcal{A}$ and the target pool $\mathcal{T}$ in epoch $e+1$ under the corresponding hypothetical action, denoted as $\mathbb{E}[B^{(\mathcal{A})}_{e+1}]$ and $\mathbb{E}[B^{(\mathcal{T})}_{e+1}])$, respectively.
\end{itemize}

\paragraphNoSkip{Reward $R$} The design of the reward function $R$ is central, as it must capture the adversary's strategic objective beyond simple profit maximization.
Our formulation differs from prior work~\cite{alpturer2024optimal,nagy2025forking} in that the agent is not trained merely to maximize its own block share, but to induce a long-term degradation of the target pool $\mathcal{T}$'s performance.
We therefore construct a composite reward function $\mathcal{R}$ evaluated at epoch boundaries, ensuring that protocol causalities, especially the finality of the RANDAO mix, are respected:
\[
R =  \beta \cdot \left(\mathbb{E}[B_{e+1}^{(\mathcal{A})}] - B_{\text{loss},e}^{(\mathcal{A})}\right) - \omega \cdot \left(\mathbb{E}[B_{e+1}^{(\mathcal{T})}] - B_{\text{loss},e}^{(\mathcal{T})}\right) + \gamma \cdot L_{\text{tail}}^{e+1}.
\]

Here, $\mathbb{E}[B_{e+1}^{(\mathcal{A})}]$ and $\mathbb{E}[B_{e+1}^{(\mathcal{T})}]$ denote the expected number of slots assigned to the adversary and the target pool $\mathcal{T}$ in the subsequent epoch, respectively.
The loss terms $B_{\text{loss},e}^{(\mathcal{A})}$ and $B_{\text{loss},e}^{(\mathcal{T})}$ capture forfeited blocks caused by withholding decisions in epoch $e$, while $L_{\text{tail}}^{e+1}$ measures whether the adversary gains influence over the tail of the next epoch.

The weighting coefficients $\{\beta,\omega,\gamma\}$ allow the framework to represent different objectives.
By tuning these parameters, we can recover classical attacks or define new ones.
When $\omega = 0$, the reward reduces to the adversary's own expected value accumulation, aligning with classical selfish-mixing objectives.
Conversely, setting $\beta = 0$ yields a griefing adversary that ignores its own economic losses and acts solely to suppress the target pool $\mathcal{T}$ while preserving strategic continuity.
Under this profile, the adversary is willing to sacrifice all of its own block value in order to minimize the target pool's performance.

\subsubsection{MEV-aware extension}
Prior analytical models usually assume homogeneous block values.
To capture high-value opportunities arising from mempool fluctuations and \gls{MEV}, we extend the environment with a reduced-form shared-opportunity model inspired by prior work~\cite{barzur2023werlman,bar2025mad}.

Specifically, we maintain a bounded system-wide pool $P_s \in \{0,\dots,C-1\}$ of pending high-value opportunities at slot $s$.
At each slot, a new opportunity arrives independently with probability $\delta$ and is added to the pool up to capacity $C$:
\[
P_{s+1} = \min(P_s + X_s, C),
\qquad X_s \sim \mathrm{Bernoulli}(\delta).
\]
Each block can consume at most one pending opportunity.
Given $F>0$ denotes the additional value of a high-\gls{MEV} opportunity, a canonical block's realized value is
\[
v_s =
\begin{cases}
1+F,& P_s>0,\\
1,& P_s=0,
\end{cases}
\]
After consumption, the pool decreases by one whenever $P_s>0$.
If a high-value block is displaced by the fork choice rule, the corresponding opportunity is returned to the pool and may be inherited by the next canonical block.
This captures the fact that many high-value opportunities persist briefly after displacement and can still be captured by later canonical proposers, rather than being exhausted by the first proposer alone.

\paragraphNoSkip{State representation}
In the MEV-aware extension, we augment the observation with value-aware state variables:
the realized value carried by the adversary's withheld private branch $V_s^{(\mathrm{priv})}$,
the realized value of the competing public segment $V_s^{(\mathrm{pub})}$,
and the current-slot value $V_s^{(\mathrm{cur})}$ (i.e., the value obtainable if the current slot were made canonical immediately).

As branches evolve, adversarial withheld blocks contribute to $V_s^{(\mathrm{priv})}$, while newly appended honest blocks contribute to $V_s^{(\mathrm{pub})}$ for as long as the corresponding public segment remains contestable.
This reflects our Ethereum-specific setting, where the adversary withholds first and forks later, while future slot values remain unobserved at decision time.

\paragraphNoSkip{Reward adaptation}
We incorporate these value-aware states into the reward through the net realized value effect on each party within the current fork window.
Let
\(
\Delta V_e^{(\mathcal{A})}
\)
and
\(
\Delta V_e^{(\mathcal{T})}
\)
denote the net value gained or lost by the adversary and the target pool, respectively, within the current epoch.
These terms capture both positive value retained through successful fork outcomes, including \gls{MEV} opportunities that are realized by displacing competing blocks, and negative value lost due to sacrifice, failed withholding, reversion, or displacement from the canonical chain.
Because future high-value opportunities are unknown at decision time, the next-epoch allocation terms remain expectation-based.
Thus, the reward is defined as
\[
R =
\beta \cdot \left(\mathbb{E}[B_{e+1}^{(\mathcal{A})}] + \Delta V_e^{(\mathcal{A})}\right)
-
\omega \cdot \left(\mathbb{E}[B_{e+1}^{(\mathcal{T})}] + \Delta V_e^{(\mathcal{T})}\right)
+
\gamma \cdot L_{\text{tail}}^{e+1}.
\]
This preserves next-epoch allocation effects while making current fork decisions value-aware.

\paragraphNoSkip{Strategy discovery}
Overall, our \gls{DRL} framework allows the adversary to learn different manipulation strategies, ranging from maximizing its own slot allocation to minimizing the victim pool's slot allocation. We evaluate these behaviors in \Cref{sec:DRL-effectiveness}.

\subsection{Application-layer manipulation}
\label{sec:application-layer-manipulation}
The empirical relation between pool performance and \gls{LST} valuation creates several possible application-layer profit mechanisms for an adversary.
At a high level, the adversary benefits from private knowledge about the timing and direction of a performance-induced price movement: it can open directional positions before the market fully incorporates the degradation, or compete for downstream opportunities that arise once prices and oracle values adjust.

We group the adversary's application-layer monetization strategies into four classes.
First, the adversary can short the target \gls{LST} through lending markets, leveraged credit accounts, or derivatives.
Second, it can use prediction markets when contracts exist for pool-specific or token-price events.
Third, if the price drop pushes existing collateralized positions below liquidation thresholds, the adversary can compete to capture liquidation opportunities.
Fourth, it can arbitrage temporary price differences across different exchanges as the \gls{LST} price adjusts.
These strategies differ in capital requirements, execution risk, and exposure to \gls{MEV} competition~\cite{syang2024mevsok}; we provide a detailed taxonomy in \Cref{sec:application-layer-monetization-details}.

In the remainder of the paper, we use {\em leveraged shorting} through lending protocols as a representative monetization channel.
This choice is conservative in the sense that it does not require winning competition for a one-block \gls{MEV} opportunity, unlike liquidation or atomic arbitrage, and it exposes the adversary to explicit frictions such as liquidation risk, slippage, and borrowing cost.
Concretely, the adversary deposits neutral collateral $C$, borrows and sells the target \gls{LST}, and later repurchases the \gls{LST} to repay the debt.
With effective \gls{LTV} ratio $\rho$ and $m$ recursive leverage rounds, this creates a notional short exposure of
$
N = \rho C \cdot \frac{1-\rho^{m+1}}{1-\rho}.
$
Ignoring frictions, a normalized price drop $\Delta_p$ yields gross payoff proportional to $N\Delta_p$.
Our evaluation in \Cref{sec:application-layer-estimation} then adds frictions that determine whether this exposure is profitable: market noise, slippage, liquidation risk, and borrowing costs.
The detailed lending-based short construction is deferred to \Cref{sec:lending-based-shorting}.

\section{Evaluation Results}
\label{section:Results}

In this section, we evaluate the effectiveness of our proposed attacks.
To avoid any impact on real-world systems, all evaluations are conducted in a simulated environment.

\subsection{Effectiveness of DRL framework}
\label{sec:DRL-effectiveness}

To evaluate whether our \gls{DRL} framework can automatically discover effective leader-election manipulation strategies, we design the following research questions.

\begin{itemize}[leftmargin=*]
\item \textbf{RQ1:} Can our framework recover theoretically optimal strategies identified in prior work?
\item \textbf{RQ2:} How do adversarial profiles ranging from pure profit-seeking to pure griefing impact the performance of the target liquid staking pool?
\item \textbf{RQ3:} How does \gls{MEV} affect the manipulation strategies learned by our framework?
\end{itemize}

Throughout the evaluation, we focus on sub-$\frac{1}{3}$ stake adversaries.
This is motivated by two considerations.
First, the adversary with higher stake shares can launch catastrophic consensus-level attacks~\cite{ethereum2025posattackdefense}.
Second, the current Ethereum staking distribution~\cite{dune2026eth2staking} shows that the largest staking entity, Lido, controls less than 30\% of the total stake.

\paragraphNoSkip{Implementation}
We built a simulated Ethereum environment using Gymnasium~\cite{towers2024gymnasium} to model Ethereum's leader election mechanism and fork choice rules.
Specifically, it includes three validator groups: the attacker, the victim, and the remaining honest validators, parameterized by their stake shares. Proposer schedules are deterministically derived from the RANDAO mix using a hash-based pseudorandom process with stake-weighted sampling.

We train the attacker policy using Maskable PPO, a variant of Proximal Policy Optimization (PPO)~\cite{schulman2017proximal,huang2020closer,stable-baselines3}.
We use Maskable PPO because our environment contains state-dependent invalid actions, and action masking allows the policy to respect protocol constraints during training while improving sample efficiency.
For each adversary-victim stake configuration, we train a separate policy on episodes generated from epoch-level proposer allocations derived from independent random seeds. Each training run lasts up to $10^6$ timesteps, and performance is evaluated through repeated simulation runs. More implementation details and hyperparameter settings are deferred to \Cref{sec:ablation-study}.

\subsubsection{RQ1: Discovering optimal strategy}
To answer RQ1, we set $\omega = 0$ so that our \gls{DRL} framework prioritizes the adversary's block share, disable the MEV-aware extension and assume homogeneous block values, aligning with self-optimization strategies studied previously~\cite{alpturer2024optimal,nagy2025forking}.

To determine the best weighting coefficients $\beta,\gamma$, we perform a grid search over an $8 \times 8$ parameter grid.
For each coefficient, we consider the values 0.1, 0.5, 1.0, 1.5, 2.0, 2.5, 3.0, and 5.0.
We set the adversary's stake share to $33.3\%$ to calibrate the weighting coefficients.
After the grid search, we find that when $\beta = 1.0, \gamma = 0.5$, the adversary achieves the highest block share, measured as the average block share over $10{,}000$ simulation runs after training.

We then vary the adversary's stake share from $1\%$ to $33\%$ in steps of $1\%$. For each stake share, we train five \gls{DRL} models with different random seeds and evaluate each learned strategy over $10{,}000$ independent randomized simulations. We measure effectiveness using the same metric as before, namely the average number of slots assigned to the adversary in the next epoch, and report the mean across the five training seeds.

\begin{figure}
    \centering
    \includegraphics[width=0.65\linewidth]{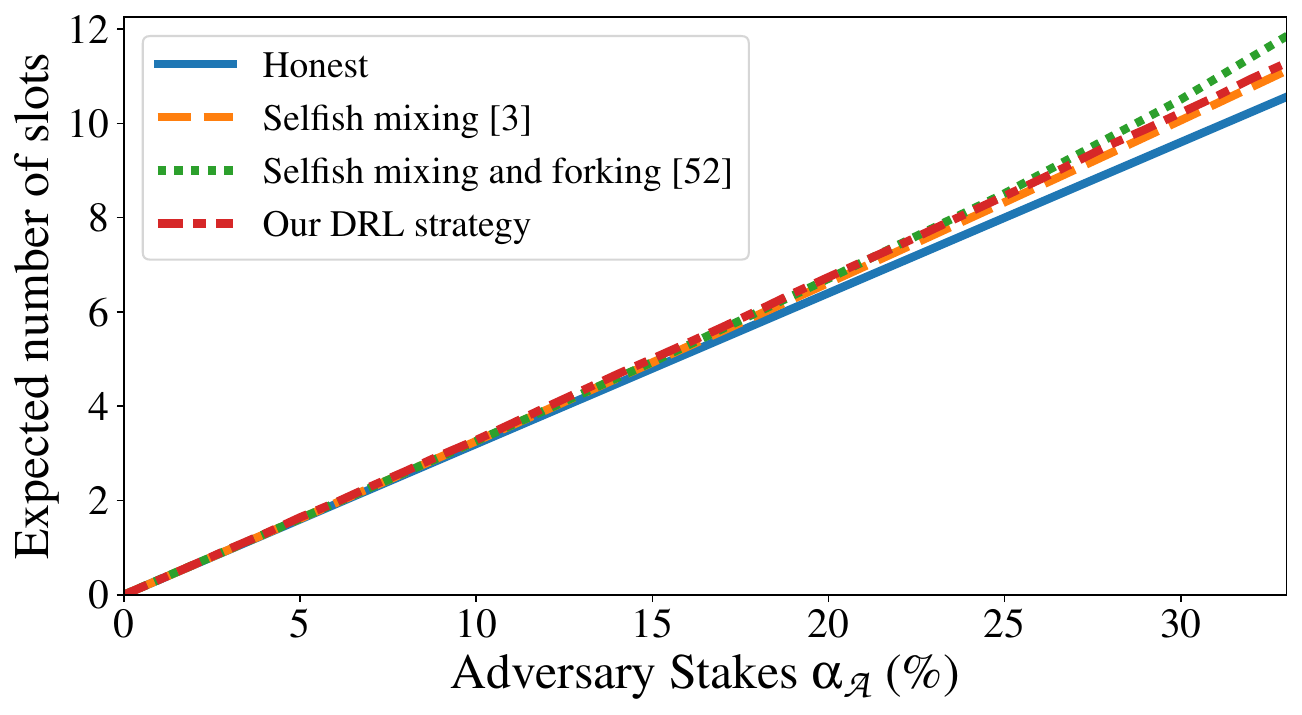}
    \captionDescription{Expected adversarial slots under different strategies and adversary stake.
    }
    \label{fig:optimizing-block-share}
\end{figure}

As \Cref{fig:optimizing-block-share} shows, we compare the effectiveness of different strategies: the honest strategy, the selfish-mixing strategy of~\cite{alpturer2024optimal} (where the adversary only chooses between honest proposing and deliberately missing), and the selfish-mixing-and-forking strategy of~\cite{nagy2025forking} (which additionally allows forking).
The strategy learned by our \gls{DRL} framework consistently outperforms selfish-mixing at all adversary stake levels.
Notably, when the stake is below $25\%$, it performs comparably to the selfish-mixing-and-forking strategy, suggesting that our framework can automatically recover near-optimal strategies from prior work.

We observe a small performance gap between our learned strategy and the analytically optimal strategy when the adversary's stake share exceeds $25\%$.
This is expected: prior work gives theoretical optima under idealized modeling assumptions, whereas our \gls{DRL} framework approximates optimal behavior through finite training in a stochastic environment, with a larger decision space at higher stake shares.
Since this experiment aims to validate whether our framework can automatically discover effective strategies rather than exactly reproduce the closed-form optimum, this small gap does not affect our main conclusion.

\begin{tcolorbox}[title=Answer to \textbf{RQ-1}, left=2pt, right=2pt, top=0pt,bottom=0pt]
Our \gls{DRL} framework can automatically recover strategies that maximize the adversary's own block allocation, with performance that remains very close to the theoretical optimum established in prior work.
\end{tcolorbox}

\subsubsection{RQ2: Pure-griefing strategy discovery}
Having validated the efficacy of our \gls{DRL} framework in replicating theoretically optimal strategies (RQ1), we now focus on its capacity for autonomous strategy discovery in a different scenario.
Specifically, whether it can automatically discover effective strategies against a particular liquid staking pool, hereafter referred to as the \textit{victim}.
To isolate this strategic effect, we continue to assume homogeneous block values and disable MEV-aware extension.

To bias our \gls{DRL} framework toward ``pure griefing,'' we set $\beta = 0$ and perform the same $8 \times 8$ grid search to determine the optimal weighting coefficients $\omega$ and $\gamma$.
Specifically, we fix the adversary's and victim's stake shares at $\alpha_{\mathcal{A}} = 20\%$ and $\alpha_{\mathcal{T}} = 20\%$, respectively, to calibrate these coefficients.
The grid search shows that $\omega = 5.0$ and $\gamma = 1.5$ minimize the victim's average number of slots assigned in the subsequent epoch.

Upon calibrating the coefficients, we extend our evaluation to a broad range of stake distributions that reflect current liquid staking pools.
Specifically, we vary both the adversary's and the victim's stake shares from $5\%$ to $30\%$ in $5\%$ increments.
This range is chosen to approximate the current distribution of major staking entities, such as Lido, which controls nearly 30\% of all stake, and Ether.fi, which controls about  5\% of the total stake.

For each stake configuration, we train a separate \gls{DRL} agent under the calibrated reward setting.
We follow the same training protocol as in RQ1, using five independent training seeds per configuration.
For each seed, we evaluate the learned policy over $10{,}000$ independent randomized simulations and report the mean of the resulting metrics across seeds.

\paragraphNoSkip{Victim loss} To quantify the adversary's impact, we use two complementary metrics: \emph{victim loss} and \emph{adversary loss}, defined as
\[
\textit{victim loss}
=
\frac{B^{(\mathcal{T})}_{e+1} - 32 \cdot \alpha_{\mathcal{T}}}{32 \cdot \alpha_{\mathcal{T}}},
\qquad
\textit{adversary loss}
=
\frac{B^{(\mathcal{A})}_{e+1} - 32 \cdot \alpha_{\mathcal{A}}}{32 \cdot \alpha_{\mathcal{A}}}.
\]
These metrics measure the reduction in victim and adversary slot allocations relative to their respective ideal allocations.
Note that lower values indicate larger reductions in allocation.

\begin{figure}[t]
    \centering
    \includegraphics[width=0.55\linewidth]{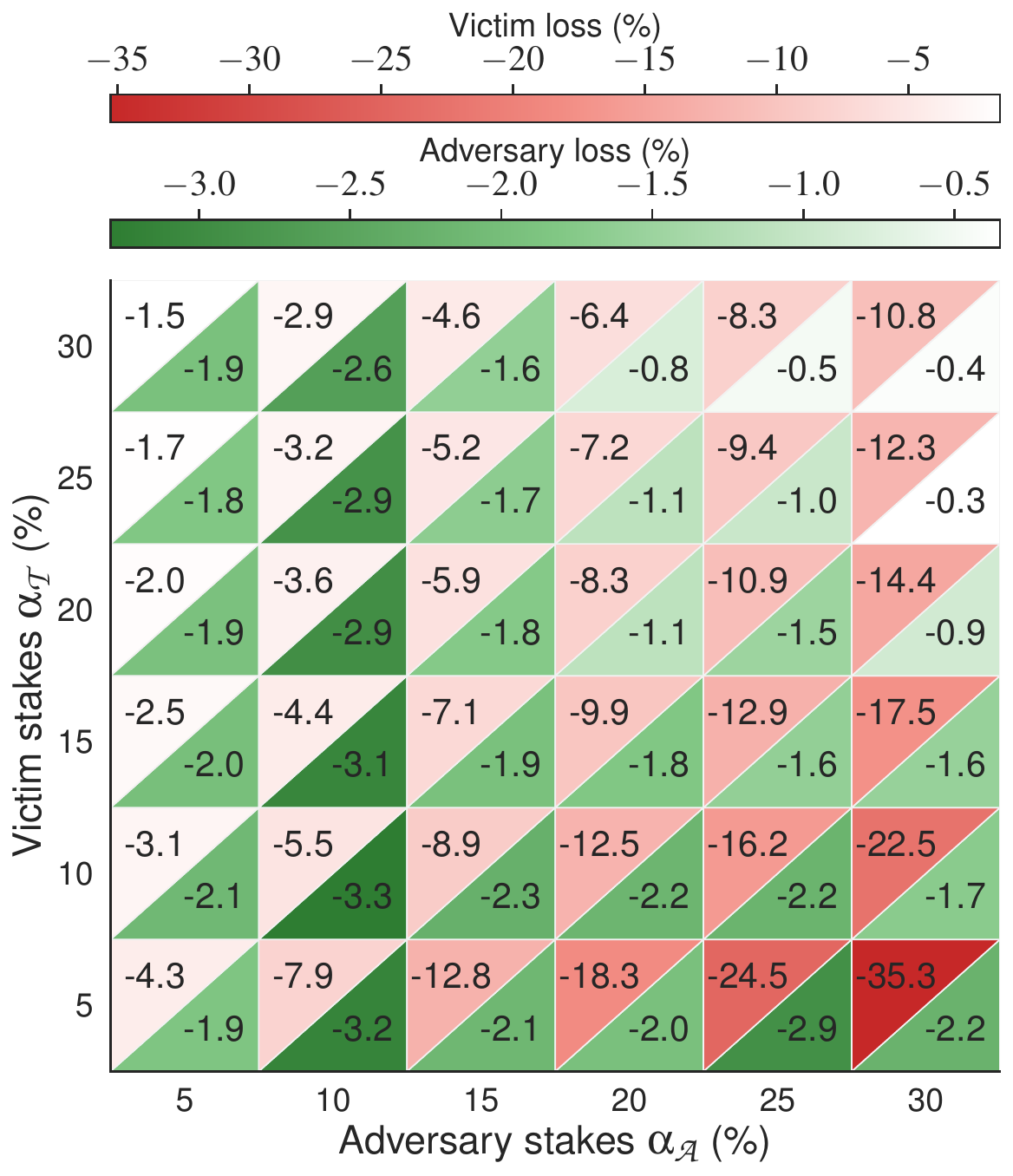}
    \captionDescription{A heatmap of the impact of learned pure griefing strategy under varying stake distributions. Each cell is divided into two triangles: the upper-left triangle (red) shows the victim's relative slot loss, while the lower-right triangle (green) shows the adversary's relative slot loss (or increase).
    More negative values are worse for the victim (red), while less negative values are better for the adversary (green).
    }
    \label{fig:profit-loss-heatmap}
\end{figure}

\Cref{fig:profit-loss-heatmap} summarizes attack outcomes across different adversary and victim stake shares.
Adversary loss is often small and negative, indicating a strategic trade-off in which the attacker incurs limited cost to amplify the victim's performance degradation.
Notably, victim loss consistently exceeds adversary loss across nearly all stake combinations.

For example, when $\alpha_{\mathcal{A}} = 10\%$ and $\alpha_{\mathcal{T}} = 30\%$, the adversary incurs only a $2.6\%$ loss while causing a $2.9\%$ loss for the victim. This shows that even a relatively small adversary can use limited self-sacrifice to amplify the victim's performance degradation. The asymmetry is even more pronounced when both the adversary and the victim control $30\%$ of the stake. In this case, the adversary suffers only a $0.4\%$ loss, whereas the victim suffers a $10.8\%$ loss.

Another observation we can make from~\Cref{fig:profit-loss-heatmap} is that the victim loss generally decreases as the victim's stake increases; however, this does not mean the attack becomes less effective. Because larger staking pools control more slots, even minor relative losses result in significant absolute damage. For instance, a fivefold increase in victim stake (from 5\% to 25\%) only yields a one-third reduction in relative loss, falling short of a proportional one-fifth improvement.
This suggests that large liquid staking pools remain vulnerable to disproportionate absolute damage from this strategy.

\paragraphNoSkip{Chain quality impact}
Pure griefing can be more damaging to chain quality than self-optimization because the adversary may accept additional block loss to reduce the victim pool's performance.
We measure this effect using {\em chain-quality impact}: the percentage of epoch slots in which an otherwise-canonical block is either sacrificed by the adversary or displaced by an adversarial fork.
Higher values indicate stronger degradation of canonical-chain utilization: more slots are wasted or replaced by the attack, reducing throughput and weakening chain-quality guarantees.

\Cref{fig:rq2-chain-quality} shows that pure-griefing strategies consistently cause higher chain-quality impact than self-optimization.
Under this objective, the learned strategy leads to a clearer degradation in chain utilization: with $30\%$ stake, more than $1.7\%$ of epoch slots are sacrificed or displaced by the adversary.
The effect is slightly stronger for smaller victim pools, but this variation is modest compared with the gap from the self-optimization baseline.

\begin{figure}
    \centering
    \includegraphics[width=0.75\linewidth]{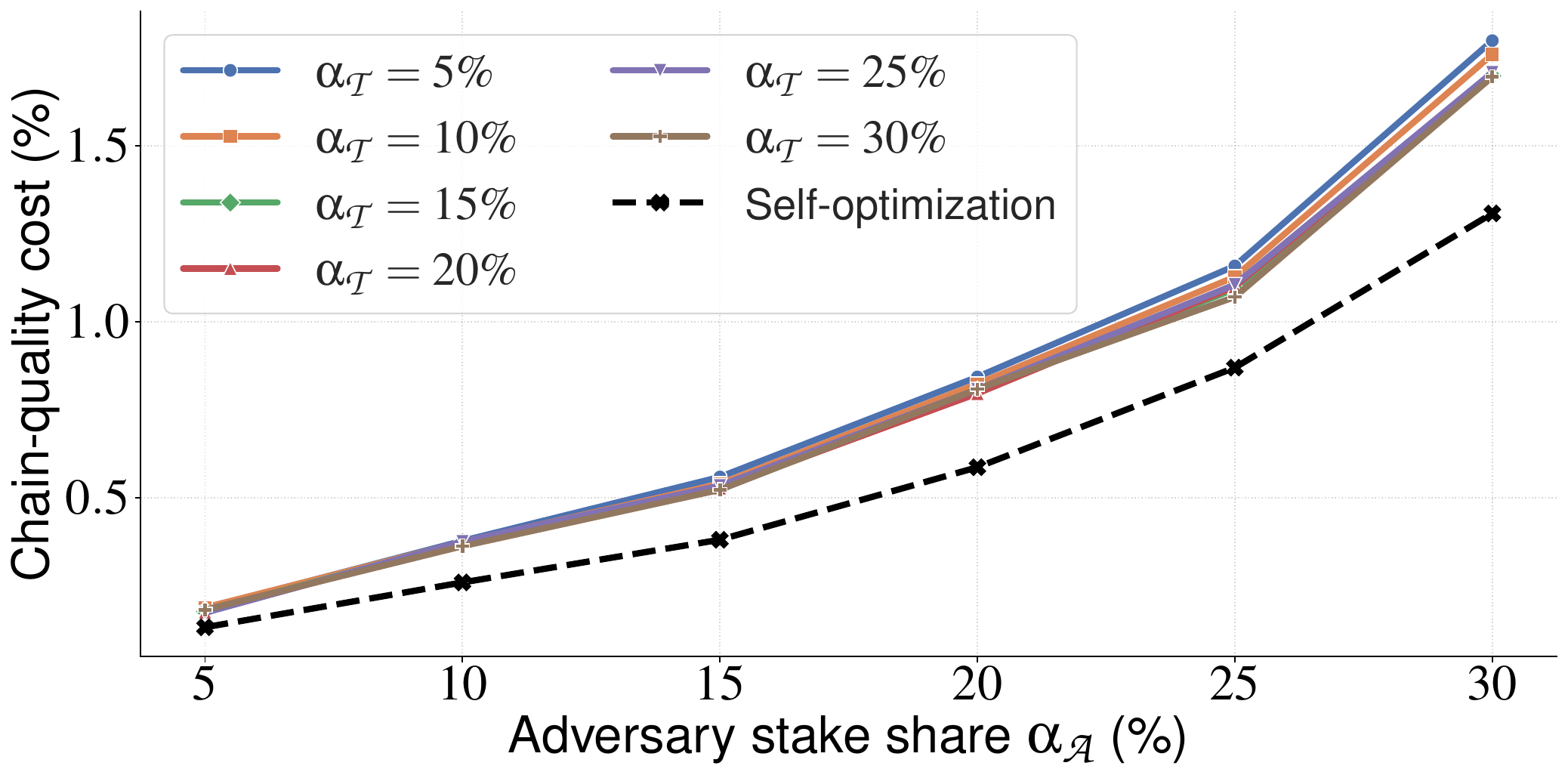}
    \captionDescription{Chain-quality impact of pure-griefing strategies.
    }
    \label{fig:rq2-chain-quality}
\end{figure}

\begin{tcolorbox}[title=Answer to \textbf{RQ-2}, left=2pt, right=2pt, top=0pt,bottom=0pt]
Our \gls{DRL} framework can also discover pure-griefing strategies that cause disproportionate victim loss with limited adversarial self-loss, while also imposing a larger chain-quality cost than self-optimization.
\end{tcolorbox}

\subsubsection{RQ3: Effect of MEV}
RQ2 shows that, under homogeneous block values, our \gls{DRL} framework can automatically discover effective griefing strategies.
We now ask whether the learned strategies differ once block values become heterogeneous due to \gls{MEV}.

To study this effect, we fix a representative adversary-victim stake pair with $\alpha_{\mathcal{A}}=\alpha_{\mathcal{T}}=20\%$ and vary the frequency and magnitude of \gls{MEV} opportunities.
This allows us to examine whether \gls{MEV} shifts the policy from pure allocation manipulation toward value-aware behavior, and how this shift affects adversary cost, victim loss, and chain-quality cost.

\paragraphNoSkip{MEV settings}
We instantiate the \gls{MEV}-aware extension by varying two parameters: the arrival probability $\delta$, which controls how often high-value opportunities enter the pool, and the bonus value $F$, which controls the additional value of a high-\gls{MEV} block.
We sweep $\delta \in \{0.05,0.10,0.15,0.20\}$ and $F \in \{1,2,5,10\}$, and fix $C=3$ as in prior work~\cite{barzur2023werlman}.

For comparability with RQ1 and RQ2, for each $(\delta,F)$ choice, we train separate \gls{DRL} agents using the calibrated self-optimization and pure-griefing reward weights from RQ1 and RQ2, and evaluate learned policies over the same number of independent simulations.

\paragraphNoSkip{Results}
We first compare the learned policies with the RQ1 and RQ2 baselines to check whether our framework recovers effective self-optimization and pure-griefing strategies under heterogeneous block values.
Due to space limits, we report the detailed validation results in \Cref{app:rq3-validation}.
The takeaway is that the learned policies achieve adversary value gains and victim losses comparable to the corresponding RQ1 and RQ2 baselines, suggesting that the framework learns effective strategies in the \gls{MEV}-aware setting.

\begin{figure}[thbp]
    \centering
    \includegraphics[width=0.9\linewidth]{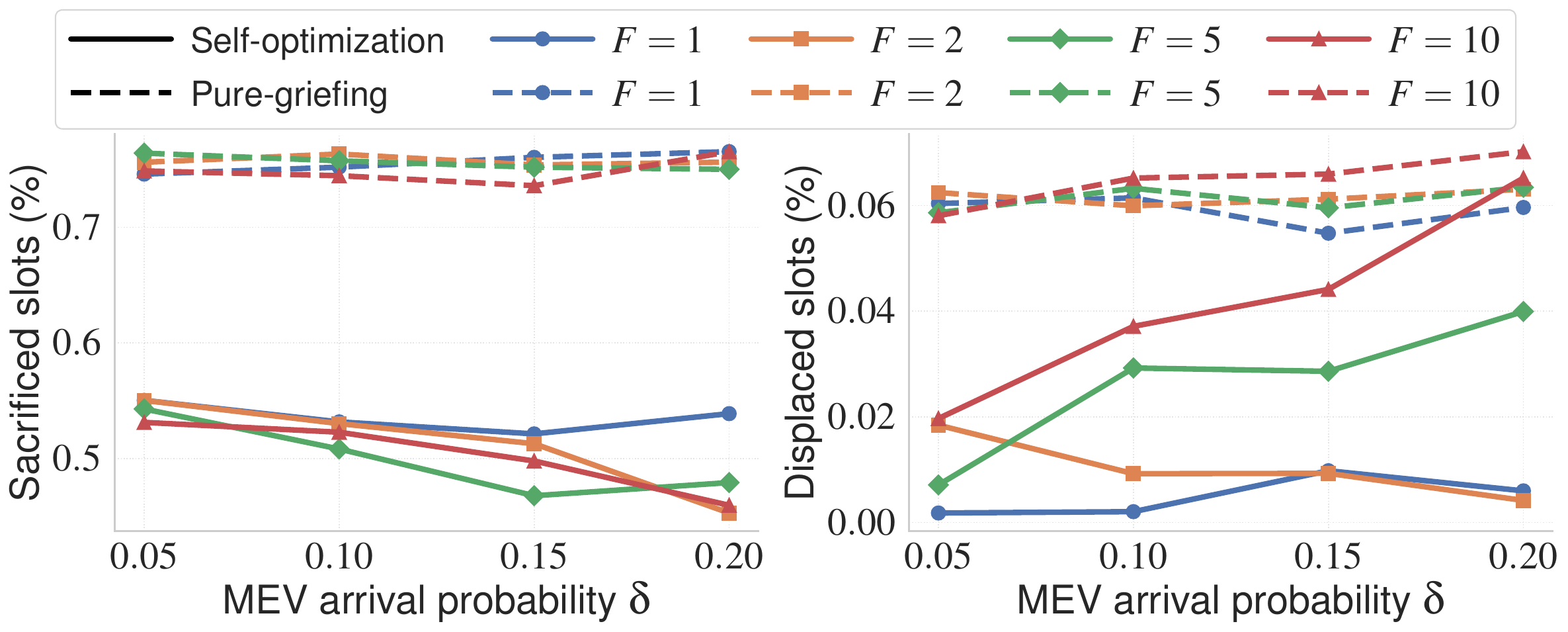}
    \captionDescription{Effect of MEV on the chain-quality cost of learned policies. The left panel reports sacrificed slots, and the right panel reports fork-displaced slots.}
    \label{fig:rq3-chain-quality-breakdown}
\end{figure}

We next examine how \gls{MEV} shifts the learned policy and how this affects chain quality.
\Cref{fig:rq3-chain-quality-breakdown} decomposes chain-quality impact into two sources: sacrificed slots and fork-displaced slots.
The two objectives respond differently to MEV.
For sacrificed slots, pure-griefing policies remain nearly unchanged across tested MEV settings: the dashed lines stay around the same level as $\delta$ and $F$ vary.
This is consistent with the objective of pure griefing, where the agent is rewarded for harming victim performance rather than preserving its own high-value proposal opportunities.
In contrast, self-optimization policies become more conservative when MEV opportunities are more valuable or frequent.
As $F$ and $\delta$ increase, the sacrificed-slot fraction generally decreases, indicating that the agent gives up fewer proposal opportunities when those opportunities are more likely to carry high MEV value.

The pattern for fork-displaced slots is different.
Displacing the victim's block can be attractive under both objectives: it can both increase the adversary's own payoff by recapturing the MEV opportunity in the displaced block and reduce the victim's rewards.
Therefore, the displaced-slot fraction tends to rise as MEV becomes stronger, that is, as $\delta$ or $F$ increases.
This increase is visible under both self-optimization and pure-griefing objectives, suggesting that MEV strengthens the incentive to fork out the victim's blocks regardless of whether the reward emphasizes the adversary's own gains or the victim's losses.

Taken together, these two components suggest that MEV has different chain-quality impacts across objectives.
Because sacrificed slots account for a large share of the chain-quality impact, stronger MEV can make self-optimization policies less harmful to chain quality by discouraging the adversary from giving up its own high-value proposal opportunities.
Under pure griefing, this self-preservation effect is absent, so stronger MEV creates additional fork opportunities that can slightly increase the overall chain-quality impact.

\begin{tcolorbox}[title=Answer to \textbf{RQ-3}, left=2pt, right=2pt, top=0pt,bottom=0pt]
\gls{MEV} makes learned policies value-aware: self-optimization policies sacrifice fewer high-value attacker slots, while both objectives fork-displace more victim blocks. Consequently, stronger \gls{MEV} can reduce chain-quality impact under self-optimization but slightly increase it under pure griefing.
\end{tcolorbox}

\subsection{Application-layer profit estimation}
\label{sec:application-layer-estimation}
Following prior work \cite{ritz2018impact,negy2020selfish,hou2021squirrl,wang2021when,cao2023leveraging,barzur2023deep,barzur2023werlman}, we adopt Monte Carlo-based simulations to estimate the adversary's potential profits via leveraged short positions.
Our goal is not to predict exact real-world profits, but to quantify whether the empirically estimated relation between pool performance and \gls{LST} price movements creates economic incentives.
We further use the simulation to assess how robust these incentives remain under realistic trading frictions and alternative parameter settings.

\paragraphNoSkip{Simulation framework}
Let $\delta^{\mathrm{APR}}$ denote the attack-induced \gls{APR}-equivalent degradation of the target pool.
Let $H$ denote the return horizon, and let $\hat{\beta}_H$ be the regression coefficient at that horizon.
Because the empirical specification in \Cref{section:EmpiricalData} relates pool \gls{APR} to forward log returns, we model the induced log-return shock as
\[
z_H = -\hat{\beta}_H \delta^{\mathrm{APR}} + \varepsilon_H,
\qquad
\varepsilon_H \sim \mathcal N(0,\sigma_H^2),
\]
where $\sigma_H$ is calibrated from the residual variation in the same regression.
We then convert this log-return shock into a normalized \gls{LST} price drop,
$\Delta_p = 1-\exp(z_H)$, which for small deviations gives
$\Delta_p \approx \hat{\beta}_H\delta^{\mathrm{APR}}-\varepsilon_H$.
Following \Cref{sec:application-layer-manipulation}, the adversary constructs a leveraged short position through a lending protocol.
Let $C$ denote the initial collateral (in ETH), $\rho$ the effective \gls{LTV} ratio, and $m$ the number of leverage amplification rounds.
The resulting short exposure is $N = \rho C(1-\rho^{m+1})/(1-\rho)$.

To account for finite market liquidity, we distinguish between the mid-price move implied by $\Delta_p$ and the execution prices at entry and exit.
Let $P_0$ denote the pre-attack \gls{LST} price and $P_1=P_0(1-\Delta_p)$ the post-attack mid price.
If the adversary shorts $q = N/P_0$ units of the \gls{LST}, the realized sell and buy-back prices are $P_{\mathrm{sell}} = P_0(1-s_{\mathrm{sell}}(q))$ and $P_{\mathrm{buy}} = P_1(1+s_{\mathrm{buy}}(q))$, where $s_{\mathrm{sell}}(q)$ and $s_{\mathrm{buy}}(q)$ are proportional slippage functions that depend on trade size and market depth.
The gross trading profit is $\Pi_{\mathrm{trade}} = q(P_{\mathrm{sell}} - P_{\mathrm{buy}})$.

As before, recursive leverage amplifies both upside and downside.
If adverse price movements trigger the protocol's liquidation threshold, we conservatively assume that the adversary may lose the entire initial collateral $C$, given the high leverage ratio and the competitive nature of \gls{MEV}-driven liquidations.
Formally, liquidation occurs when the normalized price change satisfies $\Delta_p \le \Delta_p^{\mathrm{liq}}$, where $\Delta_p^{\mathrm{liq}} < 0$ denotes the liquidation threshold corresponding to an adverse price movement, as implied by the protocol's maximum \gls{LTV} constraint and the effective leverage of the position.

We further account for borrowing costs under time-varying lending conditions.
Let $r_t$ denote the instantaneous borrowing rate of the borrowed \gls{LST} during the holding period.
For a position of size $N$ held over a duration $T$, we model the borrowing cost as $C_{\mathrm{borrow}} = N \int_0^T r_t \, dt \approx N \bar r T$, where $\bar r$ is the average realized borrowing rate over the holding window.
This formulation allows the borrowing cost to increase when utilization rises or when the short position remains open for longer.
We ignore transaction fees in the analysis, as they are second-order compared to the potential profits and losses induced by leveraged positions.
Consistent with this assumption, gas costs during our evaluation period are typically small (e.g., below \$0.2 per transaction as of April 2026~\cite{etherscan2026gastracker}). We also ignore costs at the consensus layer, as our previous evaluations show that they are close to zero in many cases.

Accordingly, the final net profit (or total loss) $\Pi$ is defined as
\[
\Pi =
\begin{cases}
- C, & \text{if liquidation occurs}, \\
\Pi_{\mathrm{trade}} - C_{\mathrm{borrow}}, & \text{otherwise}.
\end{cases}
\]

\paragraphNoSkip{Parameter instantiation and profit estimation}
Given an attack scenario characterized by $(\delta^{\mathrm{APR}}, H, \hat{\beta}_H, \sigma_H, C, \rho, m, T)$ together with a borrowing regime and a market-depth regime, we estimate the profit distribution via Monte Carlo simulation.
In each trial, we (i) sample the market noise $\varepsilon$ and compute the normalized price change $\Delta_p$, (ii) determine whether liquidation occurs by checking $\Delta_p \le \Delta_p^{\mathrm{liq}}$, (iii) compute the realized execution prices under the corresponding slippage assumptions, and (iv) compute the net profit $\Pi$ using the piecewise definition above.
We repeat this process for $10{,}000$ independent trials and report the estimated expected profit $\widehat{\mathbb{E}}[\Pi]$ and the probability of a profitable outcome $\widehat{\Pr}(\Pi>0)$.

We instantiate $\rho = 93\%$ and $\Delta_{p}^{\mathrm{liq}}=-2.15\%$ using real-world market data from the \href{https://app.aave.com/reserve-overview/?underlyingAsset=0xc02aaa39b223fe8d0a0e5c4f27ead9083c756cc2&marketName=proto_mainnet_v3}{Aave protocol}.
We set $m=6$ to model a reasonable leverage level, consider holding windows from intraday to multi-day horizons, and use borrowing costs from observed protocol rates, with $C=100$ ETH as a representative capital allocation.
For the borrowing cost, we set the annualized borrowing rate to $\bar r = 2.17\%$ according to the Aave protocol.

We capture external execution frictions through three reduced-form slippage regimes with bounds of $0.05\%$, $0.1\%$, and $0.2\%$. These values are conservative relative to slippage settings discussed for deep-liquidity swaps in Uniswap's user guidance~\cite{uniswap2026api}. If market depth is worse, the adversary can abstain from opening the position, so larger slippage mainly removes the opportunity rather than imposing an additional realized trading loss.

We set $\delta^{\mathrm{APR}} = 4\%$ as a calibrated APR-equivalent degradation.
This is based on our consensus-layer attack with $\alpha_{\mathcal{A}}=25\%$, comparable to a large staking pool such as Lido.
For victims with stake share at most $10\%$, the attack reduces realized proposer-slot allocation by more than $16\%$ (see \Cref{fig:profit-loss-heatmap}). Since block rewards are only part of total staking APR, this corresponds to an approximately $4\%$ APR-equivalent performance decline (see \Cref{sec:apr-equivalent-calibration}).

We use the full-sample 60-day estimates from \Cref{tab:lst-price-performance-60-day} as our main calibration of $(\hat{\beta}_H,\sigma_H)$.
Compared with longer horizons, this choice better matches the holding period considered in our shorting simulation, while still providing a stable estimate of the performance-price relation.
We report alternative calibrations based on other horizons and larger slippage bounds as robustness checks in \Cref{sec:robustness-simulation}.

We select Coinbase, Ether.fi, and Rocket Pool, whose stake shares span medium ($\approx 5\%$) to small ($\approx 0.5\%$), as representative examples to study profits across their \glspl{LST}.
\Cref{fig:profit-estimation-ecdf} plots the \gls{ECDF} of net profit from shorting the three \glspl{LST} under the main calibration, with a slippage bound of $0.05\%$.
For each curve, the \gls{ECDF} value at zero gives the probability of a non-profitable outcome.
Notably, Coinbase and Ether.fi exhibit broader profit distributions with large positive means, while Rocket Pool remains profitable in slightly more than half of the trials.

\begin{figure}
    \centering
    \includegraphics[width=0.65\linewidth]{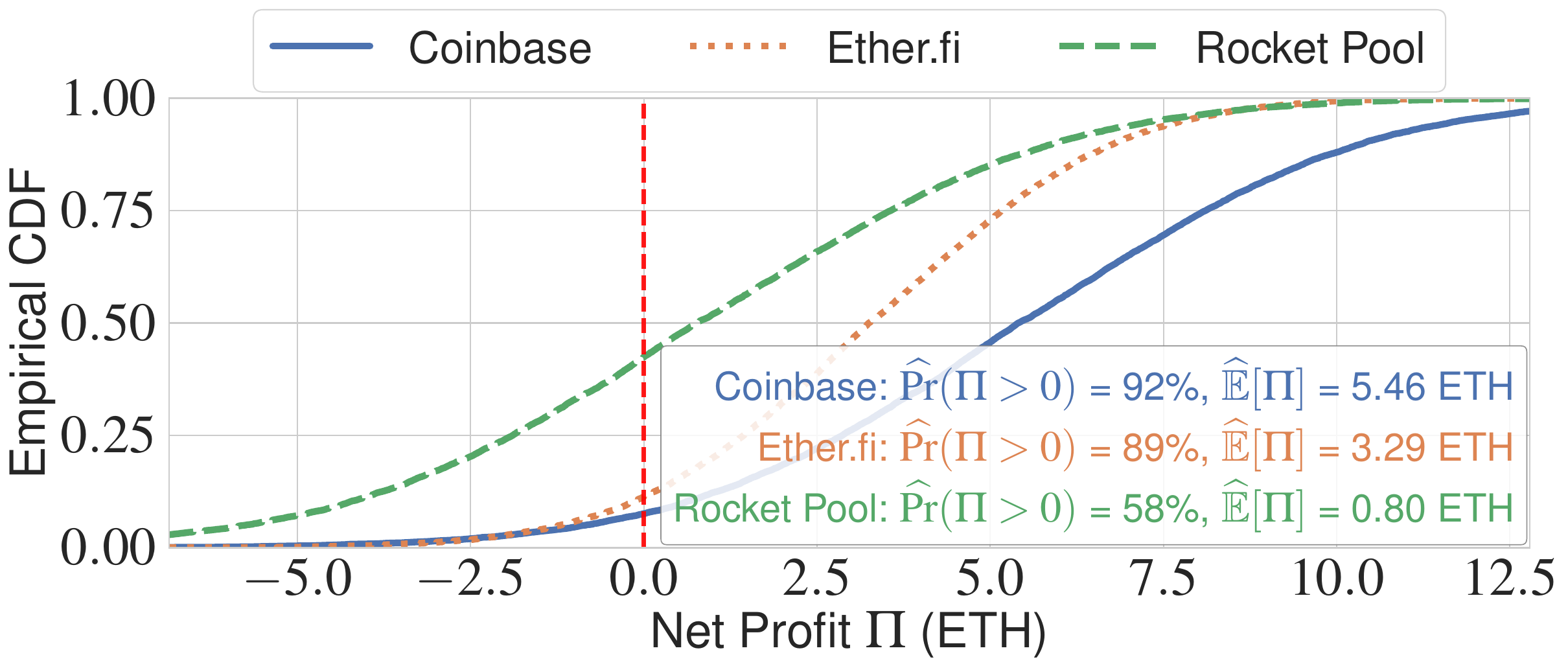}
    \captionDescription{\glspl{ECDF} of net profit from application-layer attacks on three \glspl{LST}. For a given value $x$, the ECDF reports the fraction of Monte Carlo trials with $\Pi \le x$.}
    \label{fig:profit-estimation-ecdf}
\end{figure}

Across the three \glspl{LST}, the estimated probability of profit ranges from $58\%$ to $92\%$, and the expected net profit ranges from $0.80$ ETH to $5.46$ ETH.
For an adversary allocating $C=100$ ETH, these profits are large relative to honest staking: the same capital would earn roughly $0.49$ ETH over 60 days at a $3\%$ APR.
Thus, even under a conservative slippage bound, the application-layer profit channel can dominate honest staking rewards.

\paragraphNoSkip{Break-even degradation threshold}
To compare pools' economic exposure, we compute a break-even degradation threshold $\delta^{\mathrm{BE}}_i$ for each pool $i$.
We vary the APR-equivalent performance degradation $\delta^{\mathrm{APR}}$ and define $\delta^{\mathrm{BE}}_i$ as the smallest value for which the adversary's expected profit against pool $i$ exceeds the honest staking reward over 60 days.
A lower $\delta^{\mathrm{BE}}_i$ indicates higher exposure, since a smaller consensus-layer performance loss is sufficient to make the attack economically preferable to honest participation.

Using our empirical calibration, we estimate $\delta^{\mathrm{BE}}_i$ for each pool.
Coinbase, Ether.fi, Rocket Pool, and Binance have low thresholds: $1.39\%$, $1.94\%$, $3.59\%$, and $5.35\%$, respectively, all below the APR-equivalent degradation achieved by our learned strategy.
Mantle and Lido have higher thresholds, suggesting lower exposure: $22.06\%$ and $33.38\%$, respectively.

\section{Discussion}
\label{section:Discussion}

\subsection{Mitigation}

Note that the cross-layer attacks we study represent a general vulnerability rather than a single concrete attack.
Therefore, in our mitigation analysis, we focus on defenses against different types of consensus-layer manipulation.

\paragraphNoSkip{Impossibility of preventing leader election manipulation}
We first argue that randomness beacons cannot fully eliminate leader-election manipulation without sacrificing some reasonable properties.
Before stating the formal impossibility result in \cref{res:BeaconAttackImpossibility}, we give the intuition.
The attack works by letting the adversary intentionally reduce the election chance of one account.
Any mechanism that tries to repair this effect must either reassign the lost probability mass or leave it unassigned.
Reassignment is unsafe because the beacon cannot distinguish honest accounts from additional attacker-controlled accounts.
Leaving it unassigned means that, with some probability, some slots will not have an elected leader.
This tension yields the result below.

\begin{restatable}[]{theorem}{resBeaconAttackImpossibility}
\label{res:BeaconAttackImpossibility}
Consider accounts $\left\{ 1, \dots, n \right\}$ participating in a leader election mechanism, where at least two are adversary-controlled.
Denote the probability assigned by the mechanism for account $i$ to be elected in slot $s$ by $p_i(s)$.
Suppose that after adversary-controlled account $a \in \left[n\right]$ performs a manipulation in slot $s$, the mechanism reduces $a$'s election probability for some subsequent slot $s'$ by $\Delta > 0$ without lowering the election probability of any other account: $p_a(s') = p_a(s) - \Delta < p_a(s)$ and $\forall i \ne a: p_i(s') \ge p_i(s)$.
Then, without knowing the identities of all accounts' owners, the mechanism cannot simultaneously satisfy the following properties:
(i) no adversary-controlled account has higher election probability in slot $s'$,
(ii) a leader is guaranteed to be elected in every slot: $\forall s: \sum_{i\in\left[n\right]} p_i(s) = 1$.
\end{restatable}

This suggests that robust defenses may require a combination of beacon design choices and external mechanisms, rather than standalone modifications of the randomness beacon.
Promising directions include delaying the availability of beacon outputs until after all relevant commitments have been fixed (e.g., via \glspl{VDF}~\cite{boneh2018verifiable}), and introducing stronger links between accounts and the entities that control them.

\paragraphNoSkip{Preventing other attacks}
For insider attacks, one mitigation idea is to require collateral, which can change the adversary's payoff structure.
If a node operator must lock up collateral that can be slashed when it misses proposals or attestations, then deliberately skipping slots or withholding attestations becomes more costly.
Thus, collateral can deter opportunistic insiders by making performance degradation economically unattractive.
In practice, Rocket Pool has adopted a similar idea to amplify the penalty for misbehaving node operators~\cite{rocketpool2026collateral}.

With respect to strategic forking manipulation, we observe that most existing attacks stem from specific flaws in prior protocol designs and have been addressed by subsequent protocol upgrades. For this class of attacks, the most effective mitigation is therefore either (i) for liquid staking pool operators or regular validators to deploy the latest patched protocol or (ii) for protocol designers and the community to transition to a more secure consensus mechanism~\cite{sarenche2025commitment,zhang2025available}.
The same mitigation idea also applies to network-layer or \gls{DoS} attacks if these attacks are caused or amplified by incorrect protocol design or vulnerable software implementations. We note that eclipse attacks and mempool \gls{DoS} attacks have been addressed through software upgrades~\cite{ding2024asymmetric}.
For regular \gls{DoS} attacks, standard network-layer mitigation is typically available from cloud services~\cite{cloudflare2026ddosprotection}.

\subsection{Scope and future work}

\paragraphNoSkip{\gls{DRL} approach}
As shown in \cref{sec:DRL-effectiveness}, there is a residual gap between the learned policies and the theoretical optimum in the cases with high adversary stakes. This suggests that further improvements can be made to close this gap in future work.
Moreover, the current framework models a single strategic agent on Ethereum \gls{PoS}.
This abstraction is standard in prior work \cite{eyal2014majority,sapirshtein2017optimal,gervais2016security,yaish2023uncle,barzur2023werlman,bar2025mad}; modeling multiple interacting adversaries remains largely unexplored and is a natural direction for future work.

\paragraphNoSkip{Scope of economic modeling}
Because fully modeling each \gls{LST}'s market microstructure is difficult, we use a reduced-form profit model calibrated from real historical data.
Profitability can still depend on execution conditions and parameters, especially slippage, borrowing costs, holding window $T$, and the estimated price-response parameters.
We therefore vary these parameters in robustness checks rather than relying on a single favorable setting.
Across these checks, the cross-layer channel can still create incentives larger than ordinary staking rewards in many settings.

\section{Conclusions}
In this paper, we propose a cross-layer attack: consensus manipulations can create targeted pool performance shocks which can be monetized via \gls{LST} markets.
We make the attack concrete by establishing an association between performance and \gls{LST} prices, and showing how a low-stakes adversary can manipulate the \gls{PoS} leader election mechanism to reduce the victim's realized block share and then use leveraged shorting to profit from the associated price drop.
Our evaluations show that our \gls{DRL} framework discovers strategies that significantly degrade victim performance, e.g., a 20\% adversary can reduce the realized block allocation of another 20\% pool by about 8.3\%.
Furthermore, Monte Carlo simulations reveal that such attacks are economically attractive, with a probability of profitability exceeding one-half for major staking pools and, in some cases, returns of roughly $11\times$ the honest staking \gls{APR}.
Finally, we discuss possible directions for mitigating the risks of such cross-layer attacks.

\printbibliography[heading=bibintoc]

\appendix

\section{Staking Pool Distribution and Trends}
\label{sec:staking-pool-distribution}
In this section, we provide more details about the liquid staking pool distribution. All data are sourced from an online dashboard~\cite{dune2026eth2staking}.
\Cref{figure:StakeOverTime} shows how the distribution of stake in Ethereum has evolved over time across different categories of stakers.
Liquid staking-related categories, including Liquid staking, Liquid restaking, and Staking pools, account for a large and relatively stable fraction of total stake, exceeding $50\%$ most of the time.
Note that this represents a lower bound, as some stakers labeled as ``CEXs,'' such as Binance, also provide liquid staking services and issue their \glspl{LST}~\cite{binance2026ethstaking}.
More generally, each staker is assigned a single primary label in the dashboard; this labeling scheme does not capture the full spectrum of services provided by a staker, but instead reflects its most salient category.

\begin{figure}
    \centering
    \scalebox{1}{\begin{tikzpicture}
    \begin{axis}[
        smooth,
        stack plots=y,
        area style,
        date coordinates in=x,
        xticklabel={\pgfcalendar{tickcal}{\tick}{\tick}{\pgfcalendarshorthand{m}{.}} '\shortyear{\year}},
        xtick distance=150,
        xticklabel style={rotate=90, anchor=near xticklabel},
        xmin=2022-09-12,
        xmax=2026-02-02,
        ylabel={Stake},
        legend style={at={(1.65,0.75)}, fill=white, fill opacity=1},
    ]
        \addplot [c0!60!black, fill=c0!40, opacity=0.5] table [col sep=comma, x=date, y=LiquidStaking] {figures/StakeOverTime.csv} \closedcycle;
        
        \addplot [c10!60!black, fill=c10!40, opacity=0.5] table [col sep=comma, x=date, y=LiquidRestaking] {figures/StakeOverTime.csv} \closedcycle;
        
        \addplot table [col sep=comma, x=date, y=StakingPools] {figures/StakeOverTime.csv} \closedcycle;
        
        \addplot table [col sep=comma, x=date, y=CEXs] {figures/StakeOverTime.csv} \closedcycle;
        
        \addplot [c6!60!black, fill=c6, opacity=0.5, postaction={pattern=north west lines}] table [col sep=comma, x=date, y=SoloStakers] {figures/StakeOverTime.csv} \closedcycle;
        
        \addplot [c4!60!black, fill=c4, opacity=0.5, postaction={pattern=north east lines}] table [col sep=comma, x=date, y=Unidentified] {figures/StakeOverTime.csv} \closedcycle;

        \legend{
            Liquid staking,
            Liquid restaking,
            Staking pools,
            CEXs,
            Solo stakers,
            Unidentified,
        }
    \end{axis}
\end{tikzpicture}}
    \captionDescription{
        Ethereum's stake distribution over time, according to staker category.
        Ethereum launched a preliminary version of its \gls{PoS} protocol in 2020, running it alongside its \gls{PoW} protocol until officially transitioning to \gls{PoS} on September 15th, 2022~\cite{ethereum2024beacon}.
    }
    \label{figure:StakeOverTime}
\end{figure}

\begin{table}
    \centering
    \small
    \captionDescription{The top 50 stakers currently operating on Ethereum.}
    \label{table:StakeDistribution}
    \scalebox{0.85}{\begin{filecontents*}{figures/StakeDistribution.csv}
staker,category,stake,share
Lido,Liquid Staking,8680063,23.3992
Binance,CEXs,3409104,9.1901
ether.fi,Liquid Restaking,2390809,6.445
Coinbase,CEXs,1795077,4.8391
Figment,Staking Pools,1471104,3.9657
Kraken,CEXs,1401410,3.7778
Blockdaemon,Staking Pools,989696,2.668
Everstake,Staking Pools,693280,1.8689
Kiln,Staking Pools,621410,1.6752
Rocket Pool,Liquid Staking,558621,1.5059
Bitcoin Suisse,CEXs,432856,1.1669
OKX,CEXs,419198,1.13
Upbit,CEXs,408352,1.1008
StakeWise,Staking Pools,357405,0.9635
P2P.org,Staking Pools,325710,0.878
stakefish,Staking Pools,289738,0.7811
Liquid Collective,Liquid Staking,281888,0.7599
DARMA Capital,Staking Pools,239424,0.6454
Renzo,Liquid Restaking,193281,0.521
Bitstamp,CEXs,188064,0.507
Mantle,Staking Pools,186080,0.5016
Kelp DAO,Liquid Restaking,164128,0.4425
Stader,Liquid Staking,139204,0.3753
CoinSpot,CEXs,137664,0.3711
Diva (Pre-launch),Liquid Staking,133484,0.3598
Consensys,Staking Pools,103616,0.2793
Octant,Staking Pools,100352,0.2705
Revolut,CEXs,60864,0.1641
Frax Finance,Liquid Staking,59520,0.1605
imToken,Staking Pools,55712,0.1502
Gate.io,CEXs,54816,0.1478
Bitfinex,CEXs,52099,0.1404
Cluster 0x82df,Unidentified,49376,0.1331
XHash,CEXs,41376,0.1115
StakeHound,Liquid Staking,37504,0.1011
MyEtherWallet,Staking Pools,35395,0.0954
Puffer Finance,Liquid Restaking,30341,0.0818
Bitget,CEXs,26944,0.0726
GMO Coin,CEXs,21280,0.0574
KuCoin,CEXs,20000,0.0539
Swell,Liquid Restaking,19201,0.0518
RockX,Staking Pools,18913,0.051
Mercado Bitcoin,CEXs,17888,0.0482
WEX Exchange,CEXs,16000,0.0431
Origin Protocol,Staking Pools,15520,0.0418
Cluster 0xc372,Unidentified,11968,0.0323
CoinEx,CEXs,9952,0.0268
Stkr (Ankr),Liquid Staking,9228,0.0249
Swell (Pre-launch),Liquid Staking,7723,0.0208
was.eth,Solo Stakers,7690,0.0207
\end{filecontents*}
\csvreader[
        tabular = r c c l,
        table head = \toprule \bfseries{Staker} & \bfseries{Stake (ETH)} & \bfseries{Share (\%)} & \bfseries{Category} \\ \midrule,
        late after last line = \\\bottomrule,
        head to column names,
        respect percent
    ]{figures/StakeDistribution.csv}{}
    {\staker & \stake & \share & \category}
    }
\end{table}

\Cref{table:StakeDistribution} summarizes Ethereum's top 50 stakers as of February 2026, ranked by stake.
At the top of the list, Lido accounts for close to $25\%$ of the total stake and representing the largest ETH staking provider.
Following are Binance, Ether.fi, and Coinbase, each with approximately $5\%$ to $10\%$ of the total stake.
The next tiers consist of stakers holding $1\%$ to $5\%$ and $0.5\%$ to $1\%$, representing relatively small-sized ETH stakers.
The remaining stakers fall into the tail of the stake distribution, each holding less than $0.5\%$ of the total stake.

\section{Robustness of the Price-Performance Analysis}
\label[appendix]{sec:robustness-of-price-performance-analysis}

To examine whether the relation between pool performance and \gls{LST} valuation depends on the return horizon, we repeat the same specification using 90-day and 60-day forward normalized returns.
All other definitions and estimation settings follow the main analysis.

\Cref{tab:lst-price-performance-90-day,tab:lst-price-performance-60-day} report the corresponding results.
The positive relation is relatively robust across horizons: most estimated coefficients remain positive, and Binance, Coinbase, Ether.fi, and Rocket Pool remain statistically significant under both shorter horizons.
This supports the main finding that pool performance and \gls{LST} valuation are meaningfully related, while suggesting that the relation becomes clearer over longer horizons.

\begin{table}
\centering
\captionDescription{Correlation and regression results for the relationship between pool performance and 90-day forward \gls{LST} returns.}
\small
\begin{tabular}{lccccc}
\toprule
Pool & Pearson & Spearman & $\beta$ & HAC SE & $p$-value \\
\midrule
Lido           & 0.251 & 0.181 & 0.0636 & 0.0185 & 0.0006  \\
Binance        & 0.319 & 0.365 & 0.1097 & 0.0215 & <0.0001  \\
Coinbase       & 0.662 & 0.602 & 0.5472 & 0.0731 & <0.0001 \\
Ether.fi       & 0.228 & 0.233 & 0.3877 & 0.1145 & 0.0007 \\
Rocket Pool     & 0.179 & 0.103 & 0.1296 & 0.0507 & 0.0106  \\
Mantle & 0.017 & 0.089 & 0.0155 & 0.0295 & 0.5990 \\
\bottomrule
\end{tabular}
\label{tab:lst-price-performance-90-day}
\end{table}

\begin{table}
\centering
\captionDescription{Correlation and regression results for the relationship between pool performance and 60-day forward \gls{LST} returns.}
\small
\begin{tabular}{lccccc}
\toprule
Pool & Pearson & Spearman & $\beta$ & HAC SE & $p$-value \\
\midrule
Lido           & 0.053 & 0.092 & 0.0149 & 0.0210 & 0.4788  \\
Binance        & 0.295 & 0.336 & 0.0932 & 0.0216 & <0.0001  \\
Coinbase       & 0.533 & 0.452 & 0.3587 & 0.0503 & <0.0001 \\
Ether.fi       & 0.170 & 0.204 & 0.2564 & 0.0989 & 0.0095\\
Rocket Pool     & 0.192 & 0.060 & 0.1389 & 0.0542 & 0.0103 \\
Mantle & 0.029 & 0.137 & 0.0225 & 0.0280 & 0.4208 \\
\bottomrule
\end{tabular}
\label{tab:lst-price-performance-60-day}
\end{table}

\section{Glossary}
\label[appendix]{section:Glossary}
Following is a list of the acronyms used in the paper.
\setglossarystyle{alttree}\glssetwidest{AAAA}
\printnoidxglossary[type={acronym},title={}]

\section{Other Consensus-layer Manipulations}
\label{sec:other-consensus-layer-attacks}
In this section, we provide further details on other consensus-layer manipulation approaches that are available to low-stakes adversaries.

\paragraphNoSkip{Insider attacks}
In practice, pools may allow external participants to join as node operators rather than requiring all operators to be internal to the pool~\cite{rocketpool2026collateral}.
Instead of manipulating consensus rules, the adversary can participate in the target pool as an operator and strategically miss slots or withhold attestations to degrade the pool's performance~\cite{rosenfeld2011analysis,tzinas2024principal}. 
We refer to this class of attacks as \textit{insider attacks}. 
Notably, such attacks require only standard participation privileges and do not rely on access to internal protocol implementations.

\paragraphNoSkip{Network-layer attacks}
Another way to cause the target pool's validators to miss block proposals or attestations is to disrupt their network connectivity during their assigned slots.
This requires the adversary to first identify the target's network endpoints, e.g., using the methods of~\citet{heimbach2024deanonymizing}.
Accordingly, the adversary can launch network-layer attacks to achieve this goal.
For example, eclipse attacks~\cite{heilman2015eclipse,wust2016ethereum} isolate targeted validators by controlling their network connections, causing them to observe delayed or incomplete blockchain states.
Another possibility is to launch \gls{DoS} attacks, which render the validator unavailable during its assigned slots and prevent it from performing as expected.
Examples include Mempool \gls{DoS} attacks~\cite{li2021deter,yaish2024speculative,wang2024understanding}, which can disrupt block building and propagation.

\section{Application-layer Monetization Details}
\label{sec:application-layer-monetization-details}

We now provide details for the application-layer channels summarized in \Cref{section:Attacks}.

\subsection{Taxonomy of monetization channels}
Given the established correlation between validator performance and the \gls{LST} price, an adversary can proactively trigger price deviations through consensus-layer attacks.
This capability transforms the \gls{LST} from a passive asset into a manipulatable financial instrument.
Building upon the premise of predictable price depreciation, we systematically explore the DeFi landscape to identify potential monetization channels.

\begin{table*}[htbp]
    \centering
    \captionDescription{Taxonomy of Monetization Channels}
    \label{tab:monetization-taxonomy}
    \small
\scalebox{0.75}{
\begin{tabularx}{1.3\linewidth}{l l l l l l P}
\toprule
\textbf{Channel}  & \textbf{Mechanism} & \textbf{Protocols} & \textbf{Complexity} &  \textbf{Capital}  & \textbf{Time Window} & \textbf{Key Risk}  \\
\midrule
\multirow{3}{*}{Shorting} & Plain lending & Aave, Spark & Low & High  & \multirow{3}{*}{Days to months} & \multirow{3}{=}{Liquidation, and the risk increases with leverage} \\
 & Leveraged lending & Gearbox, Morpho & Medium & Medium & \\
 & Derivatives & dYdX, Binance & Low & Low & ~ \\
\midrule
Prediction  & ~ & Polymarket, Kalshi & Low &  Low & Days to months  & Oracle Manipulation \\
\midrule
Liquidation & ~ & Aave, Spark & High & Low & 1 block (seconds) & MEV Competition  \\
\midrule
\multirow{2}{*}{Arbitrage} & Atomic & \multirow{2}{*}{Uniswap, Curve} & High & Low & 1 block (seconds)
 & \multirow{2}{*}{MEV Competition} \\
 & Non-atomic & ~ & High & High
 & Seconds to minutes & ~ \\
\bottomrule
\end{tabularx}
}
\end{table*}

As summarized in~\Cref{tab:monetization-taxonomy}, the application layer offers a diverse array of monetization channels for an adversary capable of inducing a decrease in \gls{LST} prices.
These channels vary significantly in terms of capital efficiency, risk exposure, and dependence on external market conditions.
We categorize them into four primary classes: \textit{shorting}, \textit{prediction}, \textit{liquidation}, and \textit{arbitrage}, which we analyze in detail below.

\paragraphNoSkip{Shorting} One strategy is to short the \gls{LST} associated with the target liquid staking pool.
We consider three mechanisms for executing short positions.
\begin{itemize}[leftmargin=*]
    \item One can use \textit{lending protocols} \cite{kotzer2026sok}, such as Aave~\cite{aave}.
    The adversary deposits neutral collateral (e.g., ETH) into the protocol to borrow the target \gls{LST} and then the borrowed tokens are subsequently sold on a \gls{DEX} to establish a short position.
    This approach requires upfront capital, since the position must be over-collateralized, and incurs ongoing borrowing interest as an operational cost.
    The primary risk is liquidation\footnote{Liquidation~\cite{qin2021empirical} refers to the protocol mechanism that forcibly closes undercollateralized positions, allowing liquidators to repay outstanding debt and seize collateral at a discount.}, which is triggered when a price rebound causes the debt value to exceed the protocol's liquidation threshold.
    \item \textit{Leveraging protocols} like Gearbox~\cite{gearbox}, the adversary can open ``Credit Accounts'': smart contract containers that allow for higher leverage within whitelisted DeFi protocols.
    This approach reduces the initial capital by utilizing the protocol's credit line; however, it also introduces high liquidation sensitivity, as the increased leverage makes the account highly vulnerable to even minor price fluctuations.
    \item The third mechanism is \textit{perpetual futures}. This approach uses decentralized derivatives platforms (e.g., dYdX~\cite{dydx}) or \glspl{CEX} to bet on price decreases without requiring physical delivery of the \gls{LST}. Instead of borrowing the asset, the adversary provides a collateral deposit, known as \textit{margin}, to back a synthetic short position. This mechanism imposes no strict capital floor, allowing the adversary to selectively determine their leverage ratio.
    While a higher leverage ratio enables significant capital efficiency, it proportionally increases the risk of liquidation.
\end{itemize}

\paragraphNoSkip{Prediction}
Adversaries can further leverage prediction markets~\cite{saguillo2025unravelling}, such as Polymarket~\cite{polymarket}, to monetize their internal knowledge of an impending attack.
By betting on unfavorable outcomes (e.g., a price decrease) that they intentionally induce, they exploit information asymmetry between informed adversaries and uninformed participants.
Unlike short selling, prediction markets offer a relatively low capital barrier, allowing adversaries to place bets without significant collateral requirements.
However, this channel is constrained by market depth (i.e., the availability of counterparties) and settlement risk, since the eventual payout depends on the correctness and potential manipulability of the oracle.

\paragraphNoSkip{Liquidation}
As the \gls{LST} price declines, many existing collateralized positions may hit their liquidation thresholds.
Since the adversary has the privileged knowledge regarding the timing and impact of the attack, they can preemptively prepare liquidation transactions to capture these opportunities the moment the price drops.
In terms of capital requirements, this channel is highly efficient as it can be facilitated via flash loans~\cite{qin2021attacking}.
However, the primary risk is the intense \gls{MEV} competition~\cite{daian2020flash}: while the adversary has a head start, other searchers monitor the mempool for oracle updates.
As a result, the adversary may either lose the liquidation opportunity or have to share most of the potential profit with the block builder to secure inclusion.

\paragraphNoSkip{Arbitrage}
The token price adjustment is often non-simultaneous across \glspl{CEX} and \glspl{DEX}, which creates opportunities for arbitrage, i.e., profiting from cross-market price differences.
The adversary can exploit this information asymmetry to execute {\em atomic arbitrage} across \glspl{DEX}~\cite{zhou2021just,wang2022cyclic}, {\em non-atomic arbitrage} between \glspl{CEX} and \glspl{DEX}, as well as across multiple Layer-2 blockchains~\cite{heimbach2024non,wu2025measuring,oz2025cross}.
Atomic arbitrage requires little capital due to flash loans, whereas other forms require capital on \glspl{CEX} or across chains.
Nevertheless, the primary risk mirrors that of liquidation: the \gls{MEV} competition introduces execution uncertainty and potentially prohibitive costs for the adversary.

\subsection{Shorting LST through lending protocols}
\label{sec:lending-based-shorting}
We describe a canonical mechanism for constructing a short position on an \gls{LST} via DeFi lending protocols, as this strategy is simple and can be readily executed by the adversary.

To establish a short position on an \gls{LST}, the adversary deposits a neutral asset (ETH) as collateral in a lending protocol, borrows \gls{LST}, and immediately sells the borrowed \gls{LST} on the secondary market.
Let $C$ denote the adversary's initial ETH collateral, $\rho \in (0,1)$ the effective \gls{LTV} ratio, and $P_t$ the \gls{LST}/ETH market price at time $t$.
The adversary can therefore borrow \(
Q_t = \frac{\rho C}{P_t}
\) units of \gls{LST}.
By selling the borrowed \gls{LST}, the adversary obtains $\rho C$ units of ETH.
If the \gls{LST}/ETH price subsequently decreases to $P_{t'} < P_t$, the adversary can repurchase the borrowed \gls{LST} at a lower cost of
\[
Q_t \cdot P_{t'} =  \rho C \cdot \frac{P_{t'}}{P_t}.
\]
So the adversary's gross profit from the short position is
\[
\Pi_{\text{gross}} = \rho \cdot C \cdot \Delta_p = \rho C (1 - \frac{P_{t'}}{P_t}).
\]

\paragraphNoSkip{Leveraged short}
We model leverage amplification in discrete rounds $r=0,1,\dots,k$ starting from the time $t$, where each round reuses the ETH proceeds from selling borrowed \gls{LST} as additional collateral.
At round $t$, the adversary borrows
\[
Q_r = \frac{\rho C_r}{P_r}
\]
units of \gls{LST} and immediately sells them for $\rho C_r$ units of ETH.
We assume $P_r = P_t$ within each round, i.e., the borrowing and swap operations are executed within a single short execution window, so price drift is negligible.
This can be justified by the fact that these operations can be done automatically by a script.
The adversary then deposits these ETH proceeds as collateral for the next round, yielding the recursion
\[
C_{r+1} = \rho C_r, \qquad C_0 = C.
\]
After $k$ rounds, the total borrowed (and sold) \gls{LST} amount is
\[
Q^{(k)} = \sum_{r=0}^{k} Q_r
= \frac{\rho C}{P_t}\sum_{r=0}^{k}\rho^{r}
= \frac{\rho C}{P_t}\cdot \frac{1-\rho^{k+1}}{1-\rho}.
\]
If the \gls{LST}/ETH price subsequently decreases from $P_t$ to $P_{t'}<P_{t}$, the adversary repurchases $Q^{(k)}$ units of \gls{LST} to repay the debt, yielding the gross profit
\[
\Pi^{(k)}_{\text{gross}} = Q^{(k)}(P_t-P_{t'})
= \rho C\left(1-\frac{P_{t'}}{P_{t}}\right)\cdot \frac{1-\rho^{k+1}}{1-\rho}.
\]

\section{Additional Details for DRL Evaluation}
\label{sec:ablation-study}

\subsection{Hyperparameter}

Given that the performance of \gls{PPO} is known to be sensitive to hyperparameter configurations, we employed Optuna~\cite{akiba2019optuna}, an automatic hyperparameter optimization framework, to tune the agent.
We used the Tree-structured Parzen Estimator (TPE) sampler to search for an optimal set of hyperparameters.

For RQ1, we tune the hyperparameters using the self-optimization setting with $30\%$ adversarial stake.
For RQ2, we tune them using the pure-griefing setting with a $20\%$ adversarial stake and a $20\%$ victim stake.
The search objective is to maximize the adversary's realized slot allocation for self-optimization policies, and to minimize the victim's realized slot allocation for pure-griefing policies.
For RQ3, we keep the hyperparameters aligned with the corresponding objective: the value-maximizing policies use the RQ1-tuned hyperparameters, while the victim-value-minimizing policies use the RQ2-tuned hyperparameters.
The selected hyperparameters are presented in \Cref{tab:ppo_hyperparameters}.

\begin{table*}[thbp]
\centering
\captionDescription{Optimized Hyperparameters via Optuna.}
\label{tab:ppo_hyperparameters}
\small
\scalebox{0.85}{\begin{tabular}{@{}lccc@{}}
\toprule
\textbf{Parameter} & \textbf{Search Space / Range} & \textbf{Optimized Value (RQ1)} &  \textbf{Optimized Value (RQ2)} \\ \midrule
Learning Rate & $[1 \times 10^{-5}, 1 \times 10^{-3}]$ & $1.47 \times 10^{-5}$ & $3.99 \times 10^{-5}$ \\
$n\_steps$ (Rollout Buffer) & $\{256, 512, 1024, 2048, 4096\}$ & 1024 & 256 \\
Batch Size & $\{32, 64, 128, 256, 512\}$ & 512 & 64 \\
$n\_epochs$ (Optimization) & $\{5, 10, 20\}$ & 10 & 20 \\
GAE Lambda ($\lambda$) & $[0.80, 0.98]$ & 0.929 & 0.964 \\
Clip Range ($\epsilon$) & $[0.10, 0.30]$ & 0.184 & 0.133 \\
Entropy Coef ($c_1$) & $[1 \times 10^{-5}, 1 \times 10^{-2}]$ & $1.51 \times 10^{-3}$  & $4.08\times 10^{-5}$ \\
Value Function Coef ($c_2$) & $[0.30, 1.00]$ & 0.393 & 0.409 \\
\bottomrule
\end{tabular}}
\end{table*}

\subsection{Ablation study}
We conduct an ablation study to evaluate whether our observation features help the agent learn effective strategies.
Specifically, we conduct two experiments in which both the adversary and victim stakes are 20\%, i.e., $\alpha_{\mathcal{A}}=20\%$, $\alpha_{\mathcal{T}}=20\%$.
In the full observation space, the discovered strategy achieves a 8.3\% reduction in the victim's slot allocation.
For the first experiment, we disable the counterfactual strategic oracles and train and evaluate the model using the same process as in \Cref{sec:DRL-effectiveness} (RQ2).
The results show that, in this setting, the discovered strategy can reduce only 0.8\% of the victim's slot allocation, which is significantly lower than 8.3\% in our evaluation.
In the second experiment, we disable the branch identifier and repeat the same process, with the results showing that the discovered strategy reduces 7.5\% of the victim's allocation, suggesting that this feature is also helpful.

We further evaluate the value-aware observation features in the MEV-aware setting by masking one feature at a time during evaluation while keeping the trained policy fixed. The current slot value $V_s^{(\mathrm{cur})}$ accounts for $83.26\%$ of the clipped reward drop, while $V_s^{(\mathrm{priv})}$ and $V_s^{(\mathrm{pub})}$ account for $9.70\%$ and $7.04\%$, respectively, suggesting that the policy mainly relies on current MEV value and also uses branch-level value information.

\subsection{Robustness and generalizability}
After validating the effectiveness of the discovered strategies in \Cref{sec:DRL-effectiveness}, we examine how sensitive their effectiveness is to the choice of learning algorithm and training configuration, rather than assuming that it solely reflects properties of the consensus protocol. To this end, we evaluate whether strategies learned under one algorithmic setting or stake configuration generalize across different learning algorithms and unseen stake distributions.
We repeat the RQ2 training using the A2C algorithm~\cite{mnih2016asynchronous} to assess algorithm independence. Separately, to evaluate zero-shot generalization across stake distributions, we test the fixed policy trained with $\alpha_{\mathcal{A}}=20\%$ and $\alpha_{\mathcal{T}}=20\%$ on other stake configurations without retraining.

\begin{figure}
    \centering
    \includegraphics[width=0.65\linewidth]{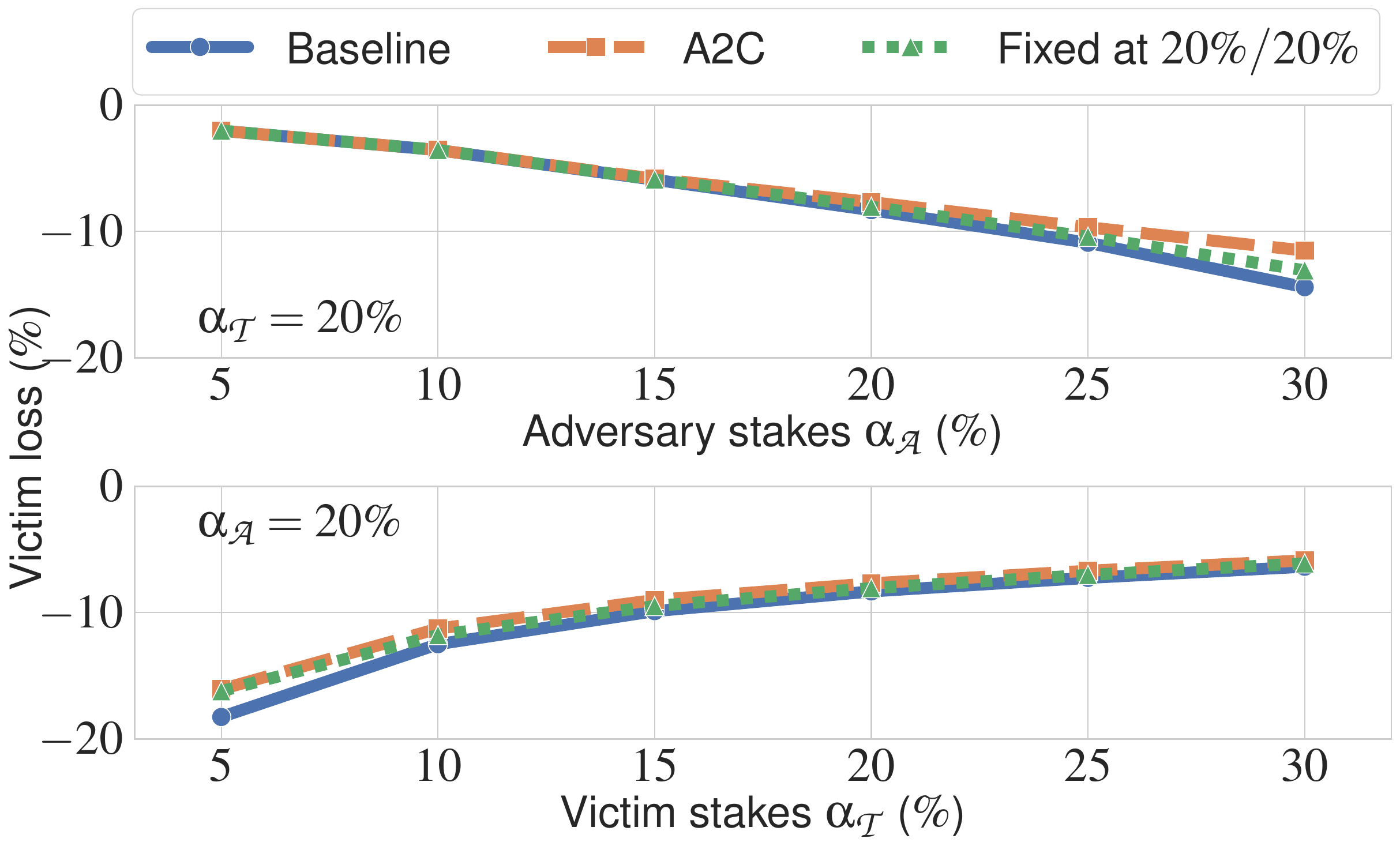}
    \captionDescription{Victim loss under three settings: Baseline (RQ2 training per stake), A2C (the same procedure using A2C), and Fixed (policy trained at $\alpha_{\mathcal{A}}=20\%$, $\alpha_{\mathcal{T}}=20\%$).}
    \label{fig:robustness-generalizability}
\end{figure}

We evaluate the effectiveness of each setting using the same metric (victim loss) as in RQ2. As shown in \Cref{fig:robustness-generalizability}, the three curves closely track each other across all stake configurations.
Even when the Fixed or A2C policies are slightly weaker than the Baseline at some stakes, e.g., $\alpha_{\mathcal{A}}=20\%, \alpha_{\mathcal{T}}=5\%$, their effectiveness in performance degradation remains stronger than the selfish-mixing strategies~\cite{alpturer2024optimal,nagy2025forking}, the estimated victim loss is about $-2\%$ under the same stake configuration.

Together, these results indicate limited dependence on the choice of learning algorithm, as well as meaningful zero-shot generalization across stake distributions.

\subsection{Policy validation under heterogeneous block values}
\label{app:rq3-validation}

In this section, we compare the adversary's gain and the victim's loss in the \gls{MEV}-aware setting with the corresponding results in RQ1 and RQ2 to assess whether our \gls{DRL} framework still learns effective strategies under heterogeneous block values.

\begin{figure}[htbp]
    \centering
    \includegraphics[width=0.5\linewidth]{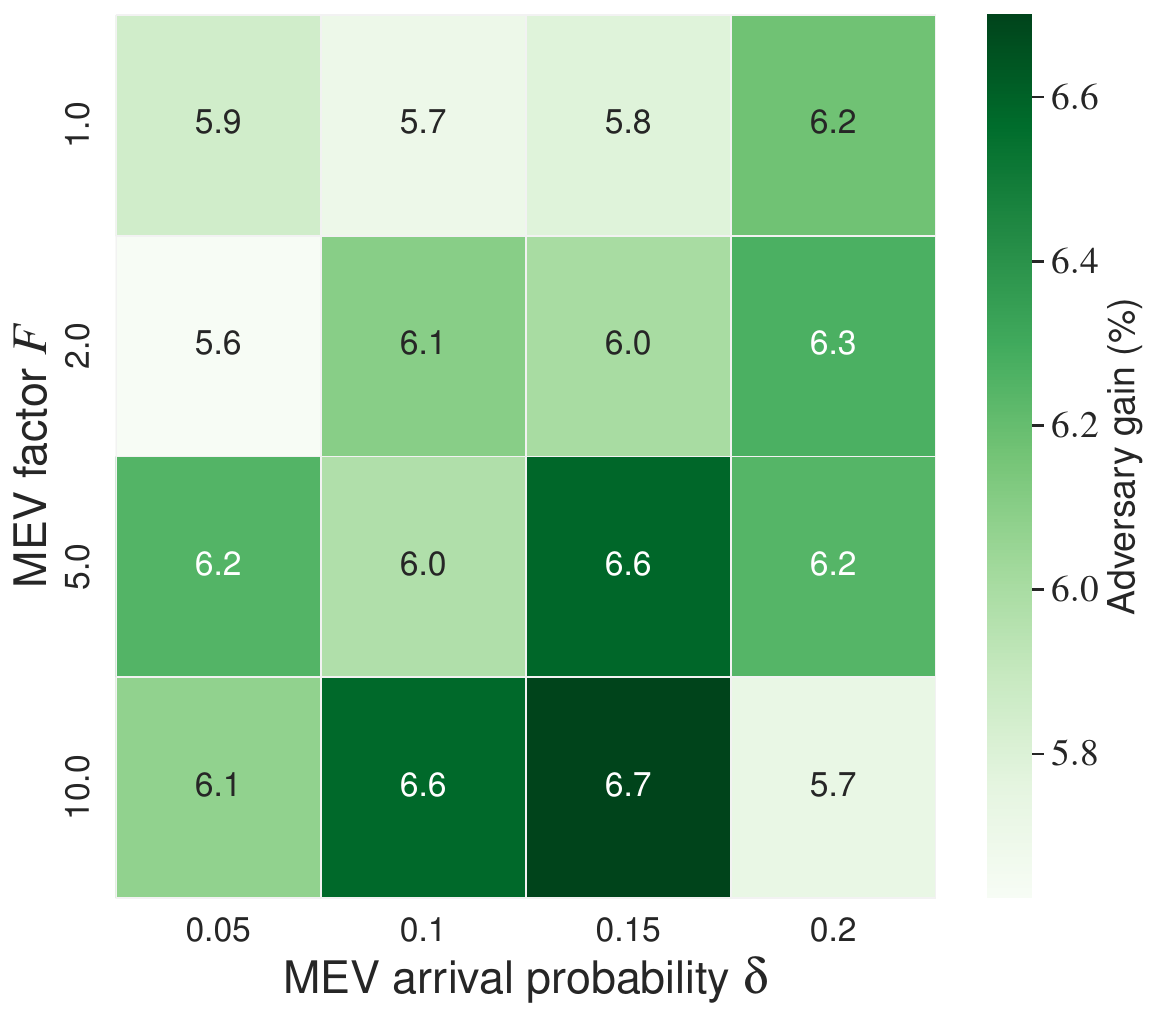}
    \captionDescription{Heatmap of the impact of the learned self-optimization strategy under varying MEV arrival probability $\delta$ and additional value $F$.}
    \label{fig:rq3-rq1-heatmap}
\end{figure}

As shown in~\Cref{fig:rq3-rq1-heatmap}, we measure the impact of the learned self-optimization strategy in terms of the adversary's value gain.
For each pair of \gls{MEV} parameters $(\delta,F)$, we compute the expected value captured by the adversary and compare it with the realized value captured during evaluation.
The results show that the adversary's value gain (about $6\%$ above the expected value) is slightly higher than in the baseline setting, where the adversary increases its slot allocation by $4.8\%$.
This suggests that our \gls{DRL} framework not only learns an effective self-optimization strategy, but also exploits heterogeneous block values by more effectively capturing high-value blocks when \gls{MEV} opportunities are present.

\begin{figure}[htbp]
    \centering
    \includegraphics[width=0.5\linewidth]{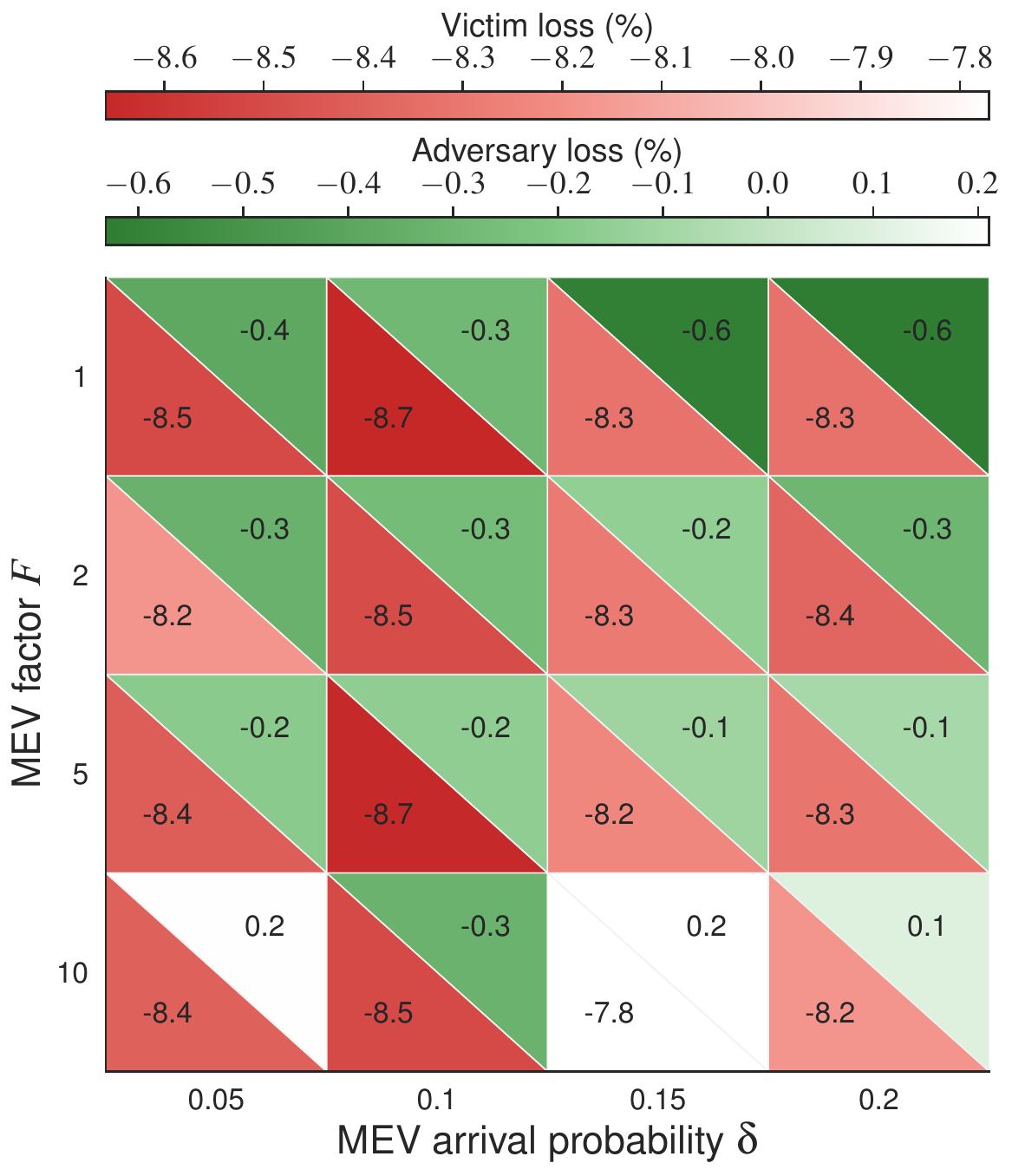}
    \captionDescription{Heatmap of the impact of the learned pure-griefing strategy under varying MEV arrival probability $\delta$ and additional value $F$.}
    \label{fig:rq3-heatmap}
\end{figure}

As shown in~\Cref{fig:rq3-heatmap}, we measure the impact of the learned pure-griefing strategy in terms of value loss.
For each pair of \gls{MEV} parameters $(\delta,F)$, we compute the expected value captured by the adversary and the victim based on the expected per-slot value induced by $\delta$ and $F$.
We then compare these baselines with the realized values captured during evaluation, and report the resulting adversary loss and victim loss.
The results show that the victim's loss (about $\-8\%$) remains close to that in RQ2 across different \gls{MEV} settings, while the adversary's own loss is even lower.
This confirms that our \gls{DRL} framework can still learn effective pure-griefing strategies under heterogeneous block values.

\section{Additional Profit Simulation Details}
\label{app:profit-simulation-details}

\subsection{APR-equivalent degradation calibration}
\label{sec:apr-equivalent-calibration}
The consensus-layer attack directly reduces the target pool's realized proposal allocation, while the application-layer simulation requires an APR-equivalent degradation parameter. We calibrate this mapping using Lido's reward decomposition on 2025-06-09. On that day, block proposal rewards, including both execution-layer rewards and consensus-layer proposer rewards, were about 167 ETH, while total staking revenue was
about 683 ETH. Thus, proposal rewards accounted for approximately
$167/683 \approx 24.5\%$ of total staking revenue. A 16\% reduction in
realized proposer allocation therefore corresponds to
\[
\delta^{\mathrm{APR}} \approx 0.16 \times \frac{167}{683} \approx 4\%,
\]
which we round to 4\% in the main simulation.

\subsection{Alternative market calibrations}
\label{sec:robustness-simulation}

To test whether the main conclusion depends on the main 60-day calibration or favorable execution conditions, we repeat the profit simulation using the 90-day estimates from \Cref{tab:lst-price-performance-90-day} and larger slippage bounds.
These checks test whether profitability is driven by the longer return horizon or by overly optimistic execution assumptions.

Under the 90-day calibration, the sensitivity coefficient $\hat{\beta}_H$ is $0.5472$ for Coinbase, $0.3877$ for Ether.fi, and $0.1296$ for Rocket Pool, with corresponding $\sigma_H$ values of $0.0080$, $0.0054$, and $0.0077$.
Under the 60-day calibration, $\hat{\beta}_H$ changes to $0.3587$, $0.2564$, and $0.1389$, with corresponding $\sigma_H$ values of $0.0073$, $0.0052$, and $0.0076$.

\begin{figure}
    \centering
    \begin{subfigure}[t]{0.9\linewidth}
        \centering
        \includegraphics[width=0.65\linewidth]{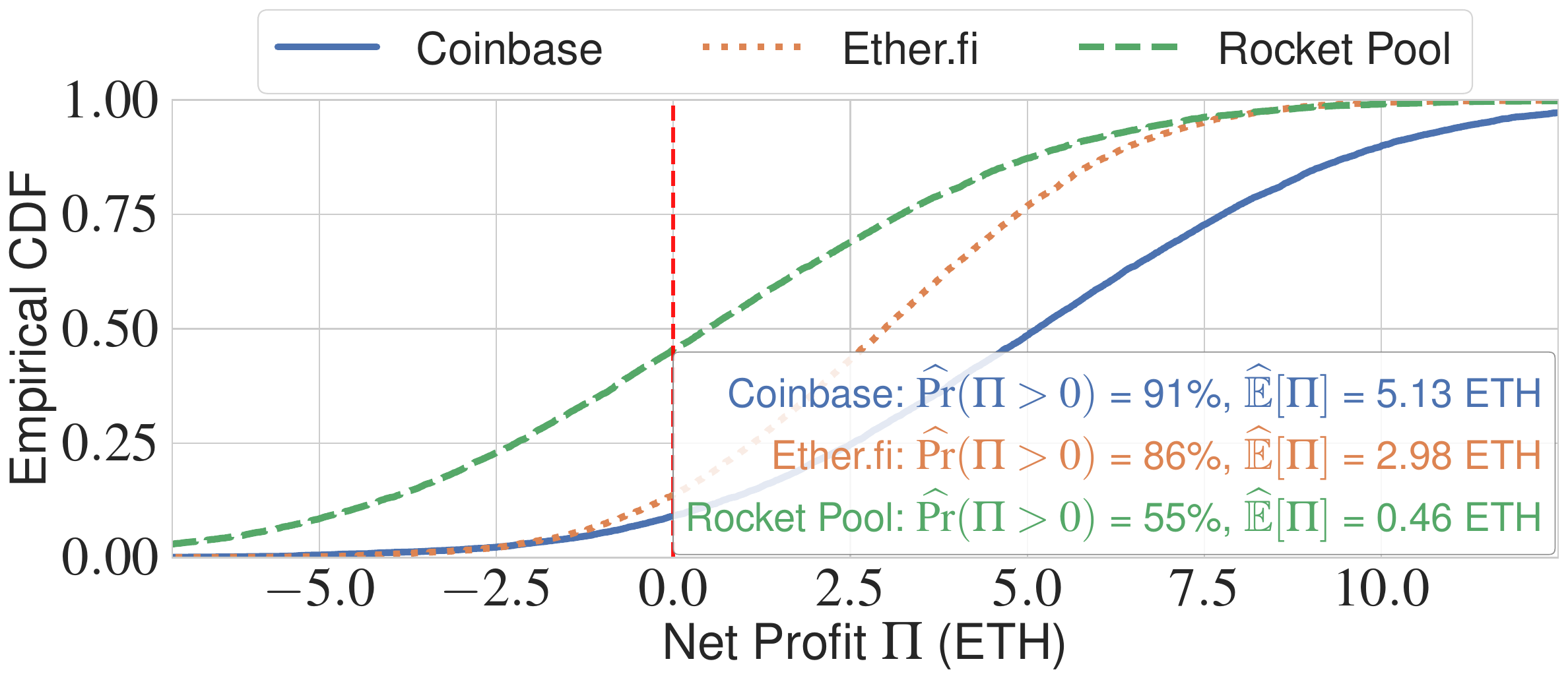}
        \captionDescription{$0.1\%$ slippage}
    \end{subfigure}

    \vspace{0.8em}

    \begin{subfigure}[t]{0.9\linewidth}
        \centering
        \includegraphics[width=0.65\linewidth]{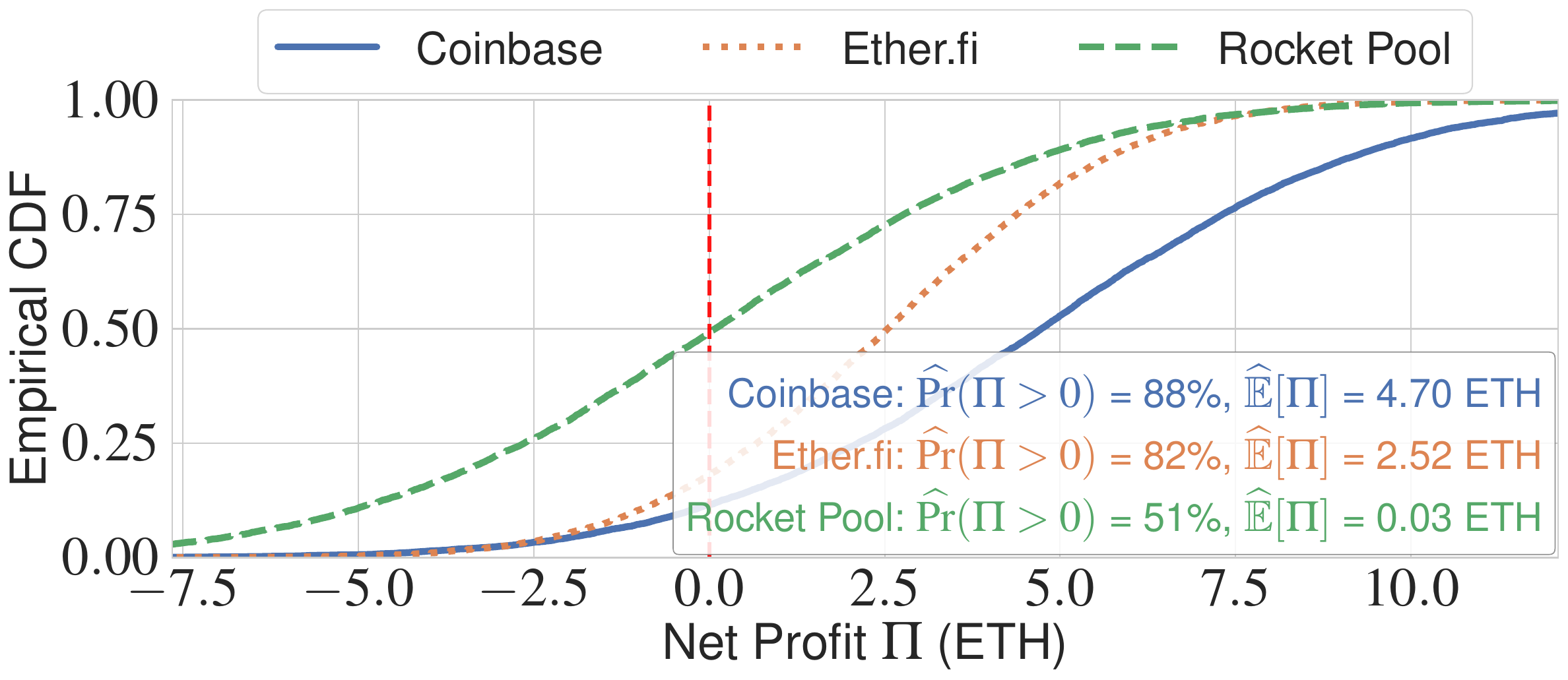}
        \captionDescription{$0.2\%$ slippage}
    \end{subfigure}

    \captionDescription{Profit distributions under higher slippage bounds using the 60-day calibration.}
    \label{fig:robustness-slippage-full-sample}
\end{figure}

\begin{figure}
    \centering
    \begin{subfigure}[t]{0.9\linewidth}
        \centering
        \includegraphics[width=0.8\linewidth]{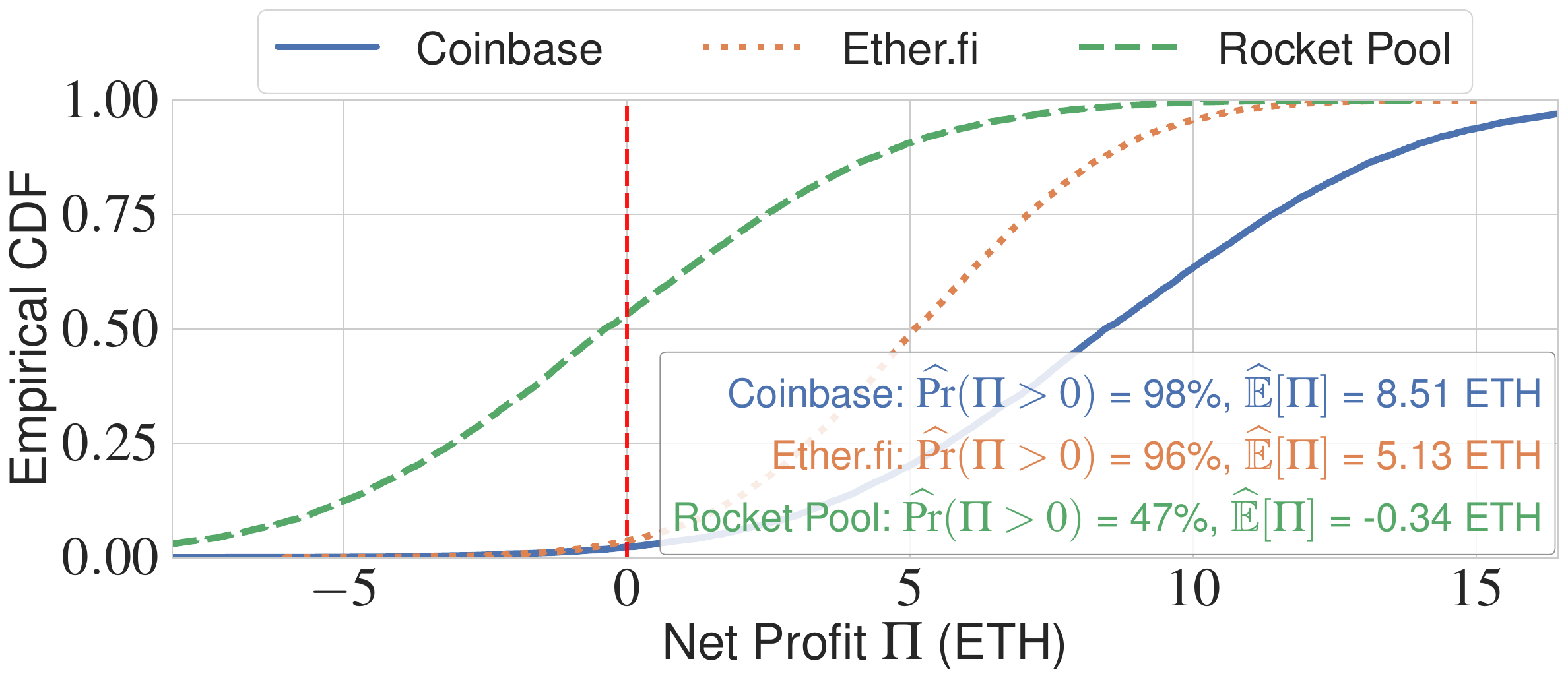}
        \captionDescription{$0.05\%$ slippage}
    \end{subfigure}
     \vspace{0.8em}
    \begin{subfigure}[t]{0.9\linewidth}
        \centering
        \includegraphics[width=0.8\linewidth]{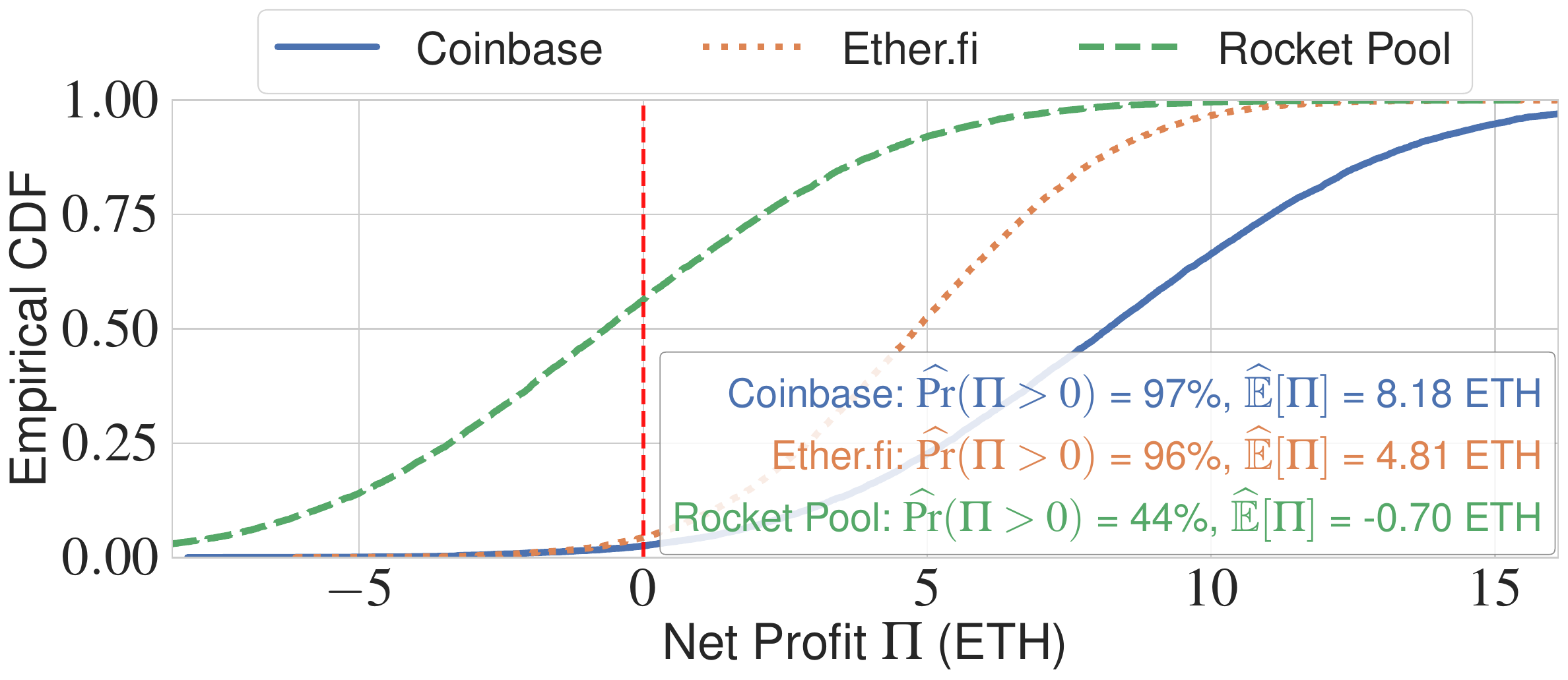}
        \captionDescription{$0.1\%$ slippage}
    \end{subfigure}
     \vspace{0.8em}
    \begin{subfigure}[t]{0.9\linewidth}
        \centering
        \includegraphics[width=0.8\linewidth]{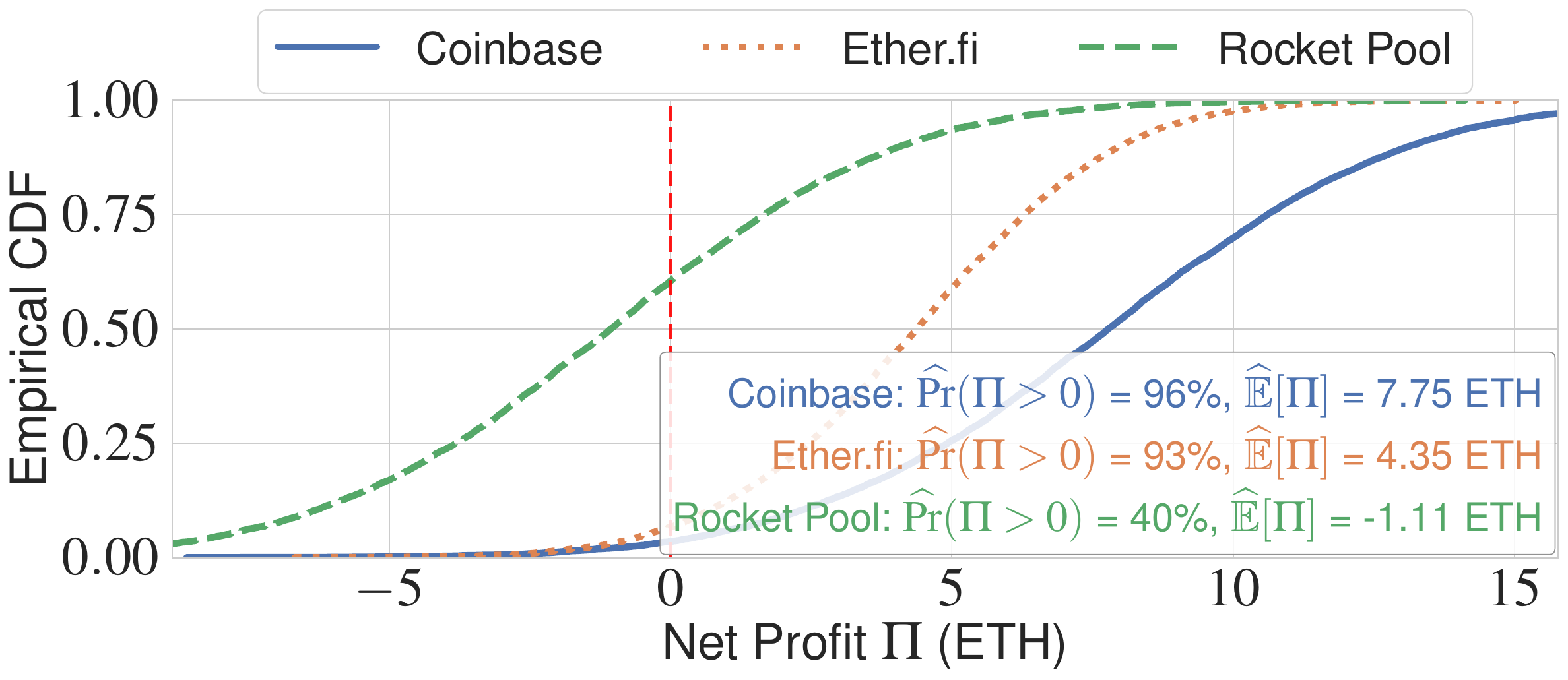}
        \captionDescription{$0.2\%$ slippage}
    \end{subfigure}
    \captionDescription{Profit distributions under the 90-day calibration across slippage bounds.}
    \label{fig:robustness-90-day-window-profit}
\end{figure}

\Cref{fig:robustness-slippage-full-sample} shows that increasing slippage shifts the profit distributions left across all three \glspl{LST}, indicating lower expected profits and lower probabilities of profitable outcomes.
However, the expected profit remains positive under both higher-slippage regimes.
Coinbase and Ether.fi remain clearly profitable, with only moderate reductions in both expected profit and profit probability.
Rocket Pool is more sensitive to execution frictions: its expected profit becomes close to zero under the $0.2\%$ slippage bound, indicating a smaller profit margin.

As shown in \Cref{fig:robustness-90-day-window-profit}, under the 90-day calibration, Coinbase and Ether.fi exhibit higher expected profits than under the 60-day baseline across all slippage regimes.
By contrast, Rocket Pool remains more sensitive to execution frictions, with negative expected profit under all three slippage bounds.
This suggests that a longer calibration horizon mainly strengthens profitability when the empirical performance--price relation is stronger, while pools with smaller margins remain vulnerable to execution frictions.

Overall, these results show that the attack is sensitive to market depth and the strength of the performance--price relation, as expected.
Nevertheless, the qualitative conclusion remains robust: across alternative calibrations and higher slippage bounds, many configurations still yield positive expected profits.
These profits often exceed honest staking rewards for the same capital, creating a strong economic incentive for the adversary.

\section{Proof of \texorpdfstring{\Cref{res:BeaconAttackImpossibility}}{Theorem~\ref{res:BeaconAttackImpossibility}}}
\label{sec:ProofBeaconAttackImpossibility}

We restate and prove \Cref{res:BeaconAttackImpossibility} below.

\resBeaconAttackImpossibility*

\begin{proof}
Assume towards contradiction that such a mechanism exists and consider two cases.

\paragraphNoSkip{There exists some account $i \neq a$ such that $p_i(s') > p_i(s)$}
Consider two executions that are identical from the mechanism's point of view and differ only in the ownership of account $i$.
In one execution, $i$ is controlled by an honest participant.
In the other, $i$ is controlled by the attacker.
By assumption, the mechanism cannot distinguish these two executions, and therefore must assign the same election probabilities in both.
It follows that any reassignment of election probability to $i$ that appears safe in the first execution also benefits the attacker in the second, contradicting property (i).

\paragraphNoSkip{There does not exist an account $i \neq a$ such that $p_i(s') > p_i(s)$}
As the mechanism is assumed to not harm any account but $a$, then we can rule out the possibility that $\exists i \neq a$ with $p_i(s') < p_i(s)$.
If so, it must be that $\forall i \ne a: p_i(s') = p_i(s)$, so we obtain:
\begin{align*}
    \sum_{i\in\left[n\right]} p_i(s')
    &
    =
    p_a(s') + \sum_{i \ne a} p_i(s')
    \\&
    =
    p_a(s) - \Delta + p_i(s)
    \\&
    =
    \left( \sum_{i\in\left[n\right]} p_i(s) \right) - \Delta
\end{align*}
By the assumption that $\Delta > 0$:
\begin{align*}
    \sum_{i\in\left[n\right]} p_i(s')
    =
    \left( \sum_{i\in\left[n\right]} p_i(s) \right) - \Delta
    <
    \sum_{i\in\left[n\right]} p_i(s)
\end{align*}
Moreover, due to property (ii):
\begin{align*}
    1
    =
    \sum_{i\in\left[n\right]} p_i(s')
    <
    \sum_{i\in\left[n\right]} p_i(s)
    =
    1
\end{align*}
A contradiction is reached, thus completing the proof.
\end{proof}

\end{document}